\documentclass[a4paper,onecolumn,superscriptaddress,11pt,accepted=2018-02-01]{quantumarticle}

\usepackage{tikz}
\usepackage{amsmath}
\usepackage{amssymb}  
\usepackage{amsthm}
\usepackage{amsfonts}
\usepackage{graphicx}
\usepackage{bbm}
\usepackage{url}

\usepackage{braket}
\usepackage{ upgreek }
\usepackage{dsfont}

\usepackage{authblk}
\usepackage[numbers]{natbib}

\usetikzlibrary{decorations.pathmorphing}
\usetikzlibrary{positioning}
\usetikzlibrary{decorations.text}

\usepackage{hyperref}
\hypersetup{
    colorlinks,
    citecolor=black,
    filecolor=black,
    linkcolor=black,
    urlcolor=black
}

\theoremstyle{plain}               
\newtheorem{thm}{Theorem}[section]
\newtheorem{lem}{Lemma}[section]
\newtheorem{cor}{Corollary}[section]

\newtheorem{defn}{Definition}[section]

\theoremstyle{remark}

\def\beq{\begin{equation}}
\def\eeq{\end{equation}}
\def\bq{\begin{quote}}
\def\eq{\end{quote}}
\def\ben{\begin{enumerate}}
\def\een{\end{enumerate}}
\def\bit{\begin{itemize}}
\def\eit{\end{itemize}}

\def\ra{\rightarrow}

\def\lb{\left(}
\def\rb{\right)}
\def\lset{\lbrace}
\def\rset{\rbrace}
\def\lk{\left\langle}
\def\rk{\right\rangle}
\def\l|{\left|}
\def\r|{\right|}
\def\lbr{\left[}
\def\rbr{\right]}

\newcommand\C{\mathbbm{C}}

\newcommand\R{\mathbbm{R}}
\newcommand\N{\mathbbm{N}}
\newcommand\M{\mathcal{M}}
\newcommand\D{\mathcal{D}}

\newcommand{\Sm}{\mathcal{S}}
\newcommand{\Em}{\mathcal{E}}

\newcommand{\Lm}{\mathcal{L}}

\newcommand{\ketbra}[1]{|#1\rangle\langle#1|}
\newcommand{\tr}{\text{tr}}
\newcommand{\one}{\mathds{1}}
\newcommand{\id}{\text{id}}
\newcommand{\liou}{\mathcal{L}}

\newcommand{\Var}{\text{Var}}
\newcommand{\Ent}{\text{Ent}}

\begin{document} 

\title{{Sandwiched R\'{e}nyi Convergence \\ for Quantum Evolutions}}

\author[1]{Alexander M\"uller-Hermes}
\email{muellerh@posteo.net}
\author[2]{Daniel Stilck Fran\c{c}a}
\email{dsfranca@mytum.de}

\affil[1]{\small{Department of Mathematical Sciences, University of Copenhagen, 2100 Copenhagen, Denmark}}

\affil[2]{\small{Department of Mathematics, Technische Universit\"at M\"unchen, 85748 Garching, Germany}}

\maketitle

\begin{abstract}
We study the speed of convergence of a primitive quantum time evolution towards its fixed point in the distance of 
sandwiched R\'{e}nyi divergences. 
For each of these distance measures the convergence is typically exponentially fast and the best exponent is 
given by a constant (similar to a logarithmic Sobolev constant) depending only on the generator of the time evolution. 
We establish relations between these constants and the logarithmic Sobolev constants as well as the spectral gap. 
An important consequence of these relations is the derivation of mixing time bounds for time evolutions directly 
from logarithmic Sobolev inequalities without relying on notions like $l_p$-regularity. 
We also derive strong converse bounds for the classical capacity of a quantum time evolution and apply these to obtain
bounds on the classical capacity of some examples, including stabilizer Hamiltonians under thermal noise.  
\end{abstract}

\newpage
\tableofcontents
\newpage

\section{Introduction}

Consider a quantum system affected by Markovian noise modeled by a quantum dynamical semigroup $T_t$ (with time parameter $t\in\R^+$) driving every initial state towards a unique full rank state $\sigma$. Using the framework of logarithmic Sobolev inequalities as introduced in \cite{Olkiewicz1999246,Kastoryanosob} the speed of the convergence towards the fixed point can be studied. 
Specifically, the $\alpha_1$-logarithmic Sobolev constant (see \cite{Olkiewicz1999246,Kastoryanosob}) is the optimal exponent $\alpha\in\R^+$ such that the inequality 
\begin{equation}
D(T_t(\rho)\|\sigma)\leq e^{-2\alpha t}D\lb\rho\|\sigma\rb
\label{equ:basicInequ}
\end{equation}
holds for the quantum Kullback-Leibler divergence, given by $D\lb\rho\|\sigma\rb=\tr\left[\rho(\ln(\rho)-\ln(\sigma))\right]$, for all $t\in \R^+$ and all states $\rho$. 

The framework of logarithmic Sobolev constants is closely linked to properties of noncommutative $l_p$-norms, and specifically to hypercontractivity~\cite{Olkiewicz1999246,Kastoryanosob}. 
Noncommutative $l_p$-norms also appeared recently in the definition of generalized R\'{e}nyi divergences (so called ``sandwiched R\'{e}nyi divergences''~\cite{muller2013quantum,strongconvrenyi}). 
It is therefore natural to study the relationship between logarithmic Sobolev inequalities and noncommutative $l_p$-norms more closely. 
The approach used here is to define constants (which we call $\beta_p$ for a parameter $p\in\lbr 1,\infty\rb$), which resemble the logarithmic Sobolev constants, but where the distance measure is a sandwiched R\'{e}nyi divergence instead 
of the quantum Kullback-Leibler divergence. More specifically, the constants $\beta_p$ will be the optimal exponents such that inequalities of the form $\eqref{equ:basicInequ}$ hold for the sandwiched R\'{e}nyi divergences $D_p$, given by
\begin{equation}
 D_p\lb\rho\|\sigma\rb=
\begin{cases} 
\frac{1}{p-1}\ln\lb\tr\left[\lb\sigma^{\frac{1-p}{2p}}\rho\sigma^{\frac{1-p}{2p}}\rb^p\right]\rb & \mbox{if }  
\text{ker}\lb\sigma\rb\subseteq\text{ker}\lb\rho\rb\text{ or }p\in \lb 0,1\rb
\\ +\infty, & \mbox{otherwise,}
\end{cases}
\end{equation}
instead of the quantum Kullback-Leibler divergence $D$.

Our main results are two-fold: 
\begin{itemize}
\item We derive inequalities between the new $\beta_p$ and other quantities such as logarithmic Sobolev constants and the spectral gap of the generator of the time evolution. These inequalities not only reveal basic properties of the $\beta_p$, but can also be used as a technical tool to strengthen results involving logarithmic Sobolev constants.
\item We apply our framework to derive bounds on the mixing time of quantum dynamical semigroups. 
Using the interplay between the $\beta_p$ and the logarithmic Sobolev constants we show how to derive a mixing time bound with the same
scaling as that of the one derived in \cite{Kastoryanosob}
directly from a logarithmic Sobolev constant.
Previously, this was only known under the additional assumption of $l_p$-regularity (see \cite{Kastoryanosob}) of the generator or for the $\alpha_1$-logarithmic Sobolev constant. 
It is still an open question whether $l_p$-regularity holds for all primitive generators.
\end{itemize}
As an additional application of our methods we derive time-dependent strong converse bounds on the classical capacity of a 
quantum dynamical semigroup. 
We apply these to some examples of systems under thermal noise. These include stabilizer Hamiltonians, such as the $2D$ toric code, and a truncated
harmonic oscillator.
To the best of our knowledge, these are the first bounds available on the classical capacity of these 
channels.
We also apply our bound to depolarizing channels, whose classical capacity is known~\cite{king2003capacity}, to benchmark our findings.

\section{Notation and Preliminaries}\label{sec:notandprelim}
Throughout this paper $\M_d$ will denote the space of $d\times d$ complex matrices.
We will denote by $\D_d$ the set of $d$-dimensional quantum states, i.e. positive semi-definite matrices $\rho\in\M_d$ with trace $1$. By $\M_d^+$ we denote the set of positive definite matrices and by $\D_d^+=\M_d^+\cap\D_d$ the set of full rank states. 

In \cite{muller2013quantum,strongconvrenyi} the following definition of sandwiched quantum R\'{e}nyi divergences was proposed:
\begin{defn}[Sandwiched $p$-R\'{e}nyi divergence]
Let $\rho,\sigma\in\D_d$. For $p\in( 0,1)\cup(1,\infty)$, the \textbf{sandwiched $p$-R\'{e}nyi divergence} is defined as:
\begin{equation}
 D_p\lb\rho\|\sigma\rb=
\begin{cases} 
\frac{1}{p-1}\ln\left(\tr\left[\lb\sigma^{\frac{1-p}{2p}}\rho\sigma^{\frac{1-p}{2p}}\rb^p\right]\right) & \mbox{if }  
\text{ker}\lb\sigma\rb\subseteq\text{ker}\lb\rho\rb\text{ or }p\in \lb 0,1\rb
\\ +\infty, & \mbox{otherwise}
\end{cases}
\label{equ:Renyi}
\end{equation}
where $\text{ker}\lb\sigma\rb$ is the kernel of $\sigma$.
\end{defn}
Note that we are using a different normalization than in \cite{muller2013quantum,strongconvrenyi}, which is more convenient
for our purposes. The logarithm in our definition is in base $e$, while theirs is in base $2$. When we write $\log$ in later sections we will mean the logarithm in base $2$.

Taking the limit $p\to1$ gives the usual quantum Kullback-Leibler divergence~\cite{umegaki1962} 
\[
\lim_{p\to1}D_p\lb\rho\|\sigma\rb=D\lb\rho\|\sigma\rb :=
\begin{cases}
\tr[\rho\lb\ln(\rho)-\ln(\sigma)\rb]  & \mbox{if }  
\text{ker}\lb\sigma\rb\subseteq\text{ker}\lb\rho\rb
\\ +\infty, & \mbox{otherwise}

\end{cases} .
\]
Similarly by taking the limit $p\ra\infty$ we obtain the max-relative entropy~\cite[Theorem 5]{muller2013quantum}
\[
\lim_{p\to\infty}D_p\lb\rho\|\sigma\rb = D_{\infty}\lb\rho\|\sigma\rb = \ln\lb\|\sigma^{-\frac{1}{2}}\rho\sigma^{-\frac{1}{2}}\|_\infty\rb.
\]
The sandwiched R\'{e}nyi divergences increase monotonically in the parameter $p\geq 1$ (see \cite[Theorem 7]{dataprocessingbeigi}) and we have
\begin{equation}
D\lb\rho\|\sigma\rb = D_1\lb\rho\|\sigma\rb\leq D_p\lb\rho\|\sigma\rb\leq D_q\lb\rho\|\sigma\rb \leq D_\infty\lb\rho\|\sigma\rb.
\label{equ:orderingSandwich}
\end{equation}
for any $q\geq p\geq 1$ and all $\rho,\sigma\in\D_d$. Next we state two simple consequences of this ordering, which will be useful later.
 
\begin{lem}\label{lem:maximalrelent}
For $\sigma\in\D_d^+$ and $p\in[1,+\infty)$
\begin{align}\label{equ:upperentropy}
\sup\limits_{\rho\in\D_d}D_p\lb \rho\|\sigma\rb=\ln\lb \|\sigma^{-1}\|_\infty\rb.
\end{align}

\end{lem}
\begin{proof}
Using \eqref{equ:orderingSandwich} for $\rho\in\D_d$ we have
\[
D_p(\rho\|\sigma)\leq D_\infty\lb\rho\|\sigma\rb = \ln\lb\|\sigma^{-\frac{1}{2}}\rho\sigma^{-\frac{1}{2}}\|_\infty\rb\leq \ln\lb\|\sigma^{-1}\|_\infty\rb.
\]
Here we used that any quantum state $\rho\in\D_d$ fulfills $\rho\leq \one_d$. Clearly, choosing $\rho = \ketbra{v_{\min}}$ for an eigenvector $\ket{v_{\min}}\in\C^d$ corresponding to the eigenvalue $\|\sigma^{-1}\|_\infty$ of $\sigma^{-1}$ achieves equality in the previous bound.

\end{proof}

%

Using \eqref{equ:orderingSandwich} together with the well-known Pinsker inequality~\cite[Theorem 3.1]{hiai1981} for the quantum Kullback-Leibler divergence we have
\begin{equation}
\frac{1}{2}\|\sigma-\rho\|_1^2\leq D\lb\rho\|\sigma\rb\leq D_p\lb\rho\|\sigma\rb
\label{equ:pinskerrenyi}
\end{equation}
for any $p\geq 1$ and all $\rho,\sigma\in\D_d$. The constant $\frac{1}{2}$ has been shown to be optimal in the classical case (see \cite{2006cs3097G}), i.e.\ restricting to $\rho$ that commute with $\sigma$, and is therefore also optimal here. 


\subsection{Noncommutative \texorpdfstring{$l_p$}{lp}-spaces}\label{sec:noncommlp}

In the following $\sigma\in\D^+_d$ will denote a full rank reference state. For $p\geq1$ we define the \textbf{noncommutative $p$-norm}  with respect to $\sigma$ as  
\begin{equation}
 \|X\|_{p,\sigma}=\lb\tr\lbr\left|\sigma^{\frac{1}{2p}}X\sigma^{\frac{1}{2p}}\right|^p\rbr\rb^{\frac{1}{p}}
\end{equation}
for any $X\in\M_d$. The space $\lb\M_d,\|\cdot\|_{p,\sigma}\rb$ is called a (weighted) noncommutative $l_p$-space. 
For a linear map $\Phi:\M_d\to\M_d$ and $p,q\geq1$ we define the \textbf{noncommutative $p\to q$-norm} with respect to $\sigma$ as
\[
\|\Phi\|_{p\to q,\sigma}=\sup\limits_{Y\in\M_d}\frac{\|\Phi(Y)\|_{q,\sigma}}{\|Y\|_{p,\sigma}}. 
\]
We introduce the \textbf{weighting operator} $\Gamma_\sigma:\M_d\ra\M_d$ as
\[
\Gamma_\sigma\lb X\rb =\sigma^{\frac{1}{2}}X\sigma^{\frac{1}{2}}.
\] 
For powers of the weighting operator we set 
\[
\Gamma^p_\sigma\lb X\rb = \sigma^{\frac{p}{2}}X\sigma^{\frac{p}{2}}
\]
for $p\in\R$ and $X\in\M_d$. We define the so called \textbf{power operator} $I_{p,q}:\M_d\ra\M_d$ as
\begin{equation}
I_{p,q}(X)=\Gamma^{-\frac{1}{p}}_\sigma\lb \left|\Gamma^{\frac{1}{q}}_\sigma\lb X\rb \right|^{\frac{q}{p}}\rb 
\label{equ:powerOp}
\end{equation}
for $X\in\M_d$. It can be verified that 
\[
\|I_{p,q}(X)\|^p_{p,\sigma} = \| X\|^q_{q,\sigma}
\]
for any $X\in\M_d$. As in the commutative theory, the noncommutative $l_2$-space turns out to be a Hilbert space, where the \textbf{weighted scalar product} is given by
\begin{equation}
\lk X,Y\rk_\sigma=\tr\lbr \Gamma_\sigma\lb X^\dagger\rb Y\rbr
\label{equ:nonCScalarp}
\end{equation}
for $X,Y\in\M_d$.
With the above notions we can express the sandwiched $p$-R\'{e}nyi divergence \eqref{equ:Renyi} for $p>1$ in terms of a noncommutative $l_p$-norm as
\begin{equation}
D_p\lb\rho\|\sigma\rb=\frac{1}{p-1}\ln\lb\| \Gamma_\sigma^{-1}\lb\rho\rb\|_{p,\sigma}^p\rb.
\label{equ:RenyiWithNorm}
\end{equation}  
For a state $\rho\in\D_d$ the positive matrix $\Gamma^{-1}_\sigma\lb\rho\rb\in\M_d$ is called the \textbf{relative density} of $\rho$ with respect to $\sigma$. 
Note that any $X\geq 0$ with $\| X\|_{1,\sigma}=1$ can be written as $X = \Gamma^{-1}_\sigma\lb\rho\rb$ for some state $\rho\in\D_d$. 
We will simply call operators $X\geq 0$ that satisfy $\| X\|_{1,\sigma}=1$ relative densities when the reference state is clear.

We refer to \cite{Olkiewicz1999246,Kastoryanosob} and references therein for proofs and more details about the concepts introduced in this section.




\subsection{Quantum dynamical semigroups}
A family of quantum channels, i.e. trace-preserving completely positive maps, $\lset T_t\rset_{t\in\R_0^+}$, $T_t:\M_d\to\M_d$, 
parametrized by a non-negative parameter $t\in\R_0^+$ is called a 
\textbf{quantum dynamical semigroup} if $T_0 = \id_d$ (the identity map in $d$ dimensions), 
$T_{t+s} = T_t\circ T_s$ for any $s,t\in\R_0^+$ and $T_t$ depends continuously on $t$. 
Any quantum dynamical semigroup can be written as $T_t = e^{t\Lm}$ (see \cite{lindblad1976,gorini}) for a Liouvillian $\liou:\M_d\ra\M_d$ of the form 
\begin{align*}
\liou(X) = \Sm(X) - \kappa X - X\kappa^\dagger,
\end{align*}
where $\kappa\in\M_d$ and $\Sm:\M_d\ra\M_d$ is completely positive such that $\Sm^{*}(\one_d) = \kappa + \kappa^\dagger$, where $\Sm^*$ is the adjoint of $\Sm$ with respect to the Hilbert-Schmidt scalar product. We will also deal with tensor powers of semigroups. For a quantum dynamical semigroup $\{T_t\}_{t\in\R^+}$ with Liouvillian $\liou$ we denote by $\liou^{(n)}$ the Liouvillian of the quantum dynamical semigroup $\{T_t^{\otimes n}\}_{t\in\R^+}$. 

In the following we will consider quantum dynamical semigroups having a full rank fixed point $\sigma\in\D^+_d$, i.e.\ the Liouvillian generating the semigroup fulfills $\liou(\sigma)=0$ (implying that $e^{t\liou}(\sigma)=\sigma$ for any time $t\in\R_0^+$). 
We call a quantum dynamical semigroup (or the Liouvillian generator) \textbf{primitive} if it has a unique full rank fixed point $\sigma$. In this case for any initial state $\rho\in\D_d$ 
we have $\rho_t = e^{t\Lm}(\rho)\ra \sigma$ as $t\ra\infty$ (see \cite[Theorem 14]{ergodicchiribella}). 

The notion of primitivity can also be defined for discrete semigroups of quantum channels. For a quantum channel $T:\M_d\to\M_d$ we will sometimes consider the discrete semigroup $\{T^{n}\}_{n\in\N}$. 
Similar to the continuous case we will call this semigroup (or the channel $T$) primitive if there is a unique full rank state $\sigma\in\D^+_d$ with
$\lim\limits_{n\to\infty}T^n(\rho)=\sigma$ for any $\rho\in\D_d$. We refer to \cite{ergodicchiribella} for other characterizations of primitive channels and sufficient conditions for primitivity.

To study the convergence of a primitive semigroup to its fixed point $\sigma$ we introduce the time evolution of the relative density $X_t = \Gamma^{-1}_{\sigma}\lb\rho_t\rb$. For any Liouvillian $\liou:\M_d\ra\M_d$ with full rank fixed point $\sigma\in\D^+_d$ define
\begin{equation}
\hat{\liou}=\Gamma^{-1}_\sigma\circ\liou\circ\Gamma_\sigma
\label{equ:HutLiou}
\end{equation}  
to be the generator of the time evolution of the relative density. Indeed it can be checked that
\[
X_t = \Gamma^{-1}_{\sigma}\lb e^{t\liou}\lb\rho\rb\rb = e^{t\hat{\liou}}\lb X\rb
\]
for any state $\rho\in\D_d$ and relative density $X = \Gamma^{-1}_\sigma\lb \rho\rb$. Note that $\|X_t\|_{1,\sigma}=\|X\|_{1,\sigma}=1$ for all $t\in\R_0^+$. Clearly the semigroup generated by $\hat{\liou}$ is completely positive and unital, but it is not trace-preserving in general. In the special case where
\begin{equation}
\Gamma^{-1}_\sigma\circ\liou\circ\Gamma_\sigma=\liou^*,
\label{equ:Reversible}
\end{equation}
the map $\hat{\liou}$ generates the adjoint of the initial semigroup, i.e. the corresponding time evolution in Heisenberg picture.
A semigroup fulfilling \eqref{equ:Reversible} is called \textbf{reversible} (or said to fulfill \textbf{detailed balance}), and in this case the Liouvillian 
$\hat{\liou}$ is a Hermitian operator w.r.t. the $\sigma$-weighted scalar product. We again refer to~\cite{Kastoryanosob,Olkiewicz1999246} for more details on these topics.
For discrete semigroups we similarly set $\hat{T}=\Gamma^{-1}_\sigma\circ T\circ\Gamma_\sigma$.

One important class of semigroups are \textbf{Davies generators}, which describe a system weakly coupled to a thermal bath under an appropriate approximation~\cite{spohn1978irreversible}. 
Describing them in detail goes beyond the scope of this article and here we will only review their most basic properties. We refer to ~\cite{davies1976quantum,breuer2007theory,DAVIES1979421} for more details. 

Suppose that we have a system of dimension $d$ weakly coupled to a thermal bath of dimension $d_B$ at inverse inverse temperature $\beta>0$. Consider a Hamiltonian $H_{\text{tot}}\in\M_d\otimes\M_{d_B}$ of the system and the bath of the form
\begin{align*}
H_{\text{tot}}=H\otimes\one_B+\one_S\otimes H_{B}+H_{I}, 
\end{align*}
where $H\in\M_d$ is the Hamiltonian of the system, $H_B\in\M_{d_B}$ of the bath and 
\begin{align}\label{equ:couplingdavies}
H_I=\sum\limits_{\alpha}S^\alpha\otimes B^\alpha\in\M_d\otimes\M_{d_B}
\end{align}
describes the interaction between the system and the bath. Here the operators $S^\alpha$ and $B^\alpha$ are self-adjoint.
Let $\{\lambda_k\}_{k\in[d]}$ be the spectrum of the Hamiltonian $H$. 
We then define the Bohr-frequencies $\omega_{i,j}$ to be given by the differences of eigenvalues of $H$, that is,
$\omega_{i,j}=\lambda_i-\lambda_j$ for different values of $\lambda$. We will drop the indices on $\omega$ from now on to avoid cumbersome notation,
as is usually done.
Moreover, we introduce operators $S^\alpha(\omega)$ which are the Fourier components of the coupling operators $S^\alpha$
and satisfy
\begin{align*}
e^{iHt}S^\alpha e^{-iHt}=\sum\limits_{\omega}S^\alpha(\omega)e^{i\omega t}. 
\end{align*}
The canonical form of the Davies generator at inverse temperature $\beta>0$ in the Heisenberg picture, $\liou_\beta^*$, is then given by
\begin{align*}
\liou_\beta^*(X)=i[H,X]+\sum\limits_{\omega,\alpha}\liou^*_{\omega,\alpha}(X), 
\end{align*}
where
\begin{align*}
\liou_{\omega,\alpha}^*(X)=G^\alpha(\omega)\lb S^\alpha(\omega)^\dagger XS^\alpha(\omega)-\frac{1}{2}\{S^\alpha(\omega)^\dagger S^\alpha(\omega),X\}\rb. 
\end{align*}
Here $\{X,Y\}=XY+YX$ is the anticommutator and $G^\alpha:\R\to\R$ are the transition rate functions.
Their form depends on the choice of the bath model~\cite{breuer2007theory}. For our purposes it will be enough to assume that these are functions that satisfy the KMS condition~\cite{Kossakowski1977}, that is,
 $G^\alpha(-\omega)=G^\alpha(\omega)e^{-\beta\omega}$.
Although this presentation of the Davies generators is admittedly very short, for our purposes it will be enough to note
that under some assumptions on the operators $S^\alpha(\omega)$~\cite{spohn,Temme2017} and on the transition rate functions, the semigroup generated by $\liou_\beta$
converges to the thermal state $\frac{e^{-\beta H}}{\tr\lb e^{-\beta H}\rb}$ and is reversible~\cite{Kossakowski1977}. In the examples considered
here this will always be the case.

\subsection{Logarithmic Sobolev inequalities and the spectral gap}

To study hypercontractive properties and convergence times of primitive quantum dynamical semigroups the framework of logarithmic Sobolev inequalities has been developed in~\cite{Olkiewicz1999246,Kastoryanosob}. 
Here we will briefly introduce this theory. For more details and proofs see \cite{Olkiewicz1999246,Kastoryanosob} and the references therein. 

We define the \textbf{operator valued relative entropy} (for $p>1$) of $X\in\M^+_d$ as
\begin{equation}
S_p(X) = -p\frac{d}{ds}I_{p+s,p}\lb X\rb|_{s=0}.
\label{equ:OpValEntr}
\end{equation}
With this we can define the $p$-relative entropy:

\begin{defn}[$p$-relative entropy]
For any full rank $\sigma\in\M^+_d$ and $p>1$ we define the \textbf{$p$-relative entropy} of $X\in\M^+_d$ as
\begin{equation}
\Ent_{p,\sigma}(X) = \lk I_{q,p}\lb X\rb,S_p(X)\rk_\sigma - \|X\|^p_{p,\sigma}\ln\lb\| X\|_{p,\sigma}\rb,
\label{equ:pRelEntr}
\end{equation}
where $\frac{1}{q}+\frac{1}{p}=1$.
For $p=1$ we can consistently define
\[
\Ent_{1,\sigma}(X)=\tr[\Gamma_\sigma(X)\lb \ln\lb\Gamma_\sigma(X)\rb-\ln(\sigma)\rb].
\]
by taking the limit $p\to1$.

\end{defn}
The $p$-relative entropy is not a divergence in the information-theoretic sense (e.g. it is not contractive under quantum channels). It was originally introduced to study hypercontractive properties of semigroups in~\cite{Olkiewicz1999246},
 where they also show it is positive for positive operators. There is however a connection to the quantum relative entropy as
\[
\Ent_{p,\sigma}\lb I_{p,1}\lb \Gamma^{-1}_\sigma\lb\rho\rb\rb\rb = \frac{1}{p} D\lb\rho\|\sigma\rb.
\] 
As a special case of the last equation we have
\[
\Ent_{1,\sigma}\lb\Gamma^{-1}_\sigma\lb\rho\rb\rb = D(\rho\|\sigma).
\] 
We may also use it to obtain an expression for $\Ent_{2,\sigma}$:
\[
\Ent_{2,\sigma}(X)=\tr\left[ \lb \Gamma_\sigma^{\frac{1}{2}} (X)\rb^2\ln\lb\frac{\Gamma_\sigma^{\frac{1}{2}} (X)}{\|X\|_{2,\sigma}}\rb\right]-\frac{1}{2}\tr\left[\lb \Gamma_\sigma^{\frac{1}{2}} (X)\rb^2\ln(\sigma)\right].
\]

We also need Dirichlet forms to define logarithmic Sobolev inequalities:

\begin{defn}[Dirichlet form]
Let $\liou:\M_d\to\M_d$ be a Liouvillian with full rank fixed point $\sigma\in\D_d^+$.
For $p>1$ we define the \textbf{$p$-Dirichlet form} of $X\in\M^+_d$ as
\[
\Em_p^{\liou}(X)=-\frac{p}{2(p-1)}\lk I_{q,p}(X),\hat{\liou}(X)\rk_{\sigma}
\]
where $\frac{1}{p}+\frac{1}{q}=1$ and $\hat{\liou}=\Gamma^{-1}_{\sigma}\circ\liou\circ\Gamma_{\sigma}$ denotes the generator of the time evolution of the relative density (cf. \eqref{equ:HutLiou}).
For $p=1$ we may take the limit $p\to1$ and consistently define the $1$-Dirichlet form by
\[
\Em_1^{\liou}\lb X\rb=-\frac{1}{2}\tr\left[\Gamma_\sigma\lb\hat{\liou}(X)\rb\lb\ln\lb\Gamma_\sigma\lb X\rb\rb-\ln(\sigma)\rb\right]. 
\]

\label{defn:DirichletForm}
\end{defn}
Formally, by making this choice we introduce the logarithmic Sobolev framework for $\hat{\liou}$ (i.e. the generator of the time-evolution of the relative density) instead of $\liou^*$. While this is a slightly different definition compared to \cite{Kastoryanosob}, where the Heisenberg picture is used, they are the same for reversible Liouvillians.  

In \cite{Olkiewicz1999246} the Dirichlet forms were introduced to study hypercontractive properties of semigroups. 
As we will see in Theorem \ref{thm:derivativesand}, they appear naturally when we compute the entropy production of the Sandwiched R\'{e}nyi divergences. 
From Corollary \ref{thm:derivativesand}
we will be able to infer that the Dirichlet form is positive for positive operators, a fact already proved in~\cite{Olkiewicz1999246}.
Both the $\Ent_{p,\sigma}$ and the Dirichlet form are intimately related to hypercontractive properties of semigroups,
as we have for a relative density $X$, some constant $\alpha>0$ and $p(t)=1+e^{2\alpha t}$ that
\begin{align*}
\frac{d}{dt}\ln\lb \|X_t\|_{p(t),\sigma}\rb=\frac{\alpha e^{\alpha t}}{\lb1+e^{\alpha t}\rb\|X_t\|_{p(t),\sigma}^{p(t)}}\lb\Ent_{p(t),\sigma}(X_t)-\frac{1}{\alpha}\Em_{p(t)}(X_t)\rb, 
\end{align*}
as shown in~\cite{Olkiewicz1999246}.

Notice that when working with $\Em_2^\liou$ we may always suppose the Liouvillian is reversible without loss of generality. This follows from the fact that
\[
\Em_2^\liou(X)=-\lk X,\hat{\liou}(X)\rk_\sigma
\]
is invariant under the additive symmetrization $\hat{\liou}\mapsto\frac{1}{2}\lb\liou^*+\Gamma^{-1}_\sigma\circ\liou\circ\Gamma_\sigma\rb$ for $X\geq0$. 

We can now introduce the logarithmic Sobolev constants:

\begin{defn}[Logarithmic Sobolev constants]
For a Liouvillian $\liou:\M_d\ra\M_d$ with full rank fixed point $\sigma\in\D^+_d$ and $p\geq1$ the $p$-\textbf{logarithmic Sobolev constant} is defined as
\begin{equation}
\alpha_p\lb\liou\rb = \sup\{\alpha\in\R^+:\hspace{0.1em}\alpha\Ent_{p,\sigma}(X)\leq\Em^\liou_p(X) \text{ for all } X>0\}
\label{equ:LogSob}
\end{equation}
\end{defn}
As $\Ent_{2,\sigma}$ does not depend on $\liou$ and, as remarked before, $\Em_2^\liou$ is invariant under an 
additive symmetrization, we may always assume without loss of generality that the Liouvillian is 
reversible when working with $\alpha_2$.

For any $X\in\M_d^+$ we can define its \textbf{variance} with respect to $\sigma\in\D_d^+$ as
\begin{equation}
\Var_\sigma\lb X\rb =\|X\|_{2,\sigma}^2-\|X\|_{1,\sigma}^2. 
\end{equation}
This defines a distance measure to study the convergence of the semigroup.
Given a Liouvillian $\liou:\M_d\ra\M_d$ with fixed point $\sigma\in\D^+_d$ we define its \textbf{spectral gap} as
\begin{equation}
\lambda(\liou) = \sup\left\{\lambda\in\R^+:\lambda\Var_{\sigma}\lb X\rb\leq\Em_2^\liou(X) \text{ for all } X>0\right\}
\label{equ:SpectralGap}
\end{equation}
where $\hat{\liou}:\M_d\ra\M_d$ is given by \eqref{equ:HutLiou}. 
We can always assume the Liouvillian to be reversible when dealing with the spectral gap, as it again depends on $\Em_2^\liou$.    

The spectral gap can be used to bound the convergence in the variance (see \cite{chisquared}), as for any $X\in\M_d^+$ we have
\begin{equation}
\frac{d}{dt}\Var_\sigma(X_t)=2\lk\hat{\liou}\lb X\rb ,X\rk_\sigma
\label{equ:DerivVar}
\end{equation}
and so
\begin{equation}\label{equ:convvariance}
\Var_\sigma\lb X_t\rb \leq e^{-2\lambda t}\Var_\sigma\lb X\rb.
\end{equation}

\section{Convergence rates for sandwiched R\'{e}nyi divergences}
\label{sec:ConvergenceRates}

In this section we consider the sandwiched R\'{e}nyi divergences of a state evolving under a primitive quantum dynamical semigroup and the fixed point of this semigroup. It is clear that these quantities converge to zero as the time-evolved state approaches the fixed point. To study the speed of this convergence we introduce a differential inequality, which can be seen as an analogue of the logarithmic Sobolev inequalities for sandwiched R\'{e}nyi divergences.

\subsection{R\'{e}nyi-entropy production}

In \cite{spohn} the entropy production for the quantum Kullback-Leibler divergence of a Liouvillian was computed. We will now derive a similar expression for the entropy production for the $p$-R\'{e}nyi divergences for $p>1$.
\begin{thm}[Derivative of the sandwiched $p$-R\'{e}nyi divergence]\label{thm:derivativesand}
 Let $\liou:\M_d\to\M_d$ be a Liouvillian with full rank fixed point $\sigma\in\D_d^+$. For any $\rho\in\D_d$ and $p>1$ we have
\begin{equation}
\frac{d}{dt}D_p(e^{t\liou}(\rho)\|\sigma)\Big{|}_{t=0}= 
\frac{p}{p-1}\frac{\tr\left[\lb\sigma^{\frac{1-p}{2p}}\rho\sigma^{\frac{1-p}{2p}}\rb^{p-1}\sigma^{\frac{1-p}{2p}}\liou\lb\rho\rb\sigma^{\frac{1-p}{2p}}\right]}
{\tr\lbr\lb\sigma^{\frac{1-p}{2p}}\rho\sigma^{\frac{1-p}{2p}}\rb^p\rbr}.
\label{equ:EntrProdLong}
\end{equation}
Using the relative density $X = \Gamma^{-1}_\sigma\lb\rho\rb$ and \eqref{equ:HutLiou} this expression can be written as:
\begin{equation}
\frac{d}{dt}D_p(e^{t\liou}(\rho)\|\sigma)\Big{|}_{t=0}=\frac{p}{p-1}\|X\|_{p,\sigma}^{-p}\lk I_{q,p}(X),\hat{\liou}\lb X\rb\rk_\sigma 
\label{equ:EntrProdShort}
\end{equation}
with $\frac{1}{p}+\frac{1}{q}=1$.
\label{thm:EntrProd}
\end{thm}
\begin{proof}

Rewriting the $p$-R\'{e}nyi divergence in terms of the relative density $X = \Gamma^{-1}_\sigma\lb\rho\rb$ and the corresponding generator $\hat{\liou}=\Gamma^{-1}_\sigma\circ\liou\circ\Gamma_\sigma$ (see \eqref{equ:HutLiou}) we have
\begin{equation}
D_p(e^{t\liou}(\rho)\|\sigma) = \frac{1}{p - 1}\ln\lb \|e^{t\hat{\liou}}\lb X\rb\|^p_{p,\sigma}\rb.
\end{equation}
By the chain rule
\[
\frac{d}{dt}D_p(e^{t\liou}\rho\|\sigma)\Big{|}_{t=0}=\frac{1}{p-1}\|X\|^{-p}_{p,\sigma}\lb\frac{d}{dt}\|e^{t\hat{\liou}}(X)\|^p_{p,\sigma}\rb\Big{|}_{t=0} .
\]
Define the curve $\gamma:\R_0^+\to\M_d$ as $\gamma(t)=\sigma^{\frac{1}{2p}}e^{t\hat{\liou}}\lb X\rb\sigma^{\frac{1}{2p}}$ and observe that
\[
\|e^{t\hat{\liou}}\lb X\rb\|^p_{p,\sigma}=\tr[\gamma(t)^p] .
\]
As the differential of the function $X\mapsto X^p$ at $A\in\M_d^+$ is given by $p A^{p-1}$, another application of the chain rule yields
\[
\frac{d}{dt}\|e^{t\hat{\liou}}(X)\|^p_{p,\sigma}\Big{|}_{t=0}=p\lk \gamma(0)^{p-1},\frac{d\gamma}{dt}(0)\rk .
\]
It is easy to check that $\frac{d\gamma}{dt}(0)=\sigma^{\frac{1}{2p}}\hat{\liou}\lb X\rb\sigma^{\frac{1}{2p}}$. 
Inserting this in the above equations and writing it in terms of the power operator \eqref{equ:powerOp} we finally obtain
\[
\frac{d}{dt}D_p(e^{t\liou}(\rho)\|\sigma)\Big{|}_{t=0}=\frac{p}{p-1}\|X\|_{p,\sigma}^{-p}\lk I_{q,p}(X),\hat{\liou}\lb X\rb\rk_\sigma 
\]
with $\frac{1}{p}+\frac{1}{q}=1$. Expanding this formula gives \eqref{equ:EntrProdLong}.

\end{proof}

By recognizing the $p$-Dirichlet form in the previous theorem we get:

\begin{cor}\label{cor:dirandentprod}
 Let $\liou:\M_d\to\M_d$ be a Liouvillian with full rank fixed point $\sigma\in\D_d^+$. For any $\rho\in\D_d$ and $p>1$ we have
\begin{equation}
\frac{d}{dt}D_p(e^{t\liou}(\rho)\|\sigma)\Big{|}_{t=0}=-2\|X\|_{p,\sigma}^{-p}\Em^\liou_p(X)\leq0, 
\end{equation}
where we used the relative density $X = \Gamma^{-1}_\sigma\lb\rho\rb$.
\label{cor:entrProd}
\end{cor}
As we remarked before, Corollary \ref{cor:dirandentprod} implies that the Dirichlet form is always positive for relative densities. 
To see this, recall that the divergences contract under quantum channels~\cite{dataprocessingbeigi} and therefore we have 
that $\frac{d}{dt}D_p(e^{t\liou}(\rho)\|\sigma)\Big{|}_{t=0}\leq0$. As $\Em^\liou_p(\lambda X)=\lambda^p\Em^\liou_p(X)$ for $\lambda>0$, this shows that it
is positive for all positive operators by properly normalizing $X$.

\subsection{Sandwiched R\'{e}nyi convergence rates}

For any $p>1$ we introduce the functional $\kappa_p:\M_d^+\ra\R$ as
\begin{equation}
\kappa_p\lb X\rb=\frac{1}{p-1}\|X\|_{p,\sigma}^p\ln\lb\frac{\|X\|_{p,\sigma}^p}{\|X\|_{1,\sigma}^p}\rb
 \label{equ:kappa}
\end{equation} 
for $X\in\M^+_d$. For $p=1$ we may again take the limit $p\to1$ and obtain $\kappa_1(X):=\lim_{p\to1}\kappa_p(X)=\Ent_{1,\sigma}(X)$. Note that $\kappa_p$ is well-defined and non-negative as $\|X\|_{p,\sigma}\geq\|X\|_{1,\sigma}$ for $p\geq 1$. Strictly speaking the definition also depends on a reference state $\sigma\in\D_d^+$, which we usually omit as it is always the fixed point of the primitive Liouvillian under consideration.

Given a Liouvillian $\liou:\M_d\ra\M_d$ with full rank fixed point $\sigma\in\D^+_d$ it is a simple consequence of Corollary \ref{cor:entrProd} that for $\rho\not=\sigma$
\begin{equation}
\frac{\frac{d}{dt}D_p(e^{t\liou}(\rho)\|\sigma)\Big{|}_{t=0}}{D_p\lb\rho\|\sigma\rb}=-2\frac{\Em_{p}^\liou(X)}{\kappa_p(X)}, 
 \label{equ:quotient}
\end{equation}
where we used the relative density $X=\Gamma^{-1}_\sigma\lb\rho\rb$, which fulfills $\|X\|_{1,\sigma}=1$. This motivates the following definition.

\begin{defn}[Entropic convergence constant for $p$-R\'{e}nyi divergence]
For any primitive Liouvillian $\liou:\M_d\ra\M_d$ and $p\geq1$ we define 
\begin{equation}
\beta_p(\liou)=\sup\{\beta\in\R^+:\hspace{0.1em}\beta\kappa_p(X)\leq\Em_{p}^\liou(X)\text{ for all } X>0\}.
\label{equ:betaSup}
\end{equation}
\end{defn}
Note that as a special case we have $\alpha_1(\liou)=\beta_1(\liou)$. It should be also emphasized that the supremum in the previous definition goes over any positive definite $X\in\M^+_d$ and not only over relative densities. However, it is easy to see that we can equivalently write
\begin{equation}
\beta_p(\liou) = \inf\Big{\lbrace} \frac{\Em_{p}^\liou(X)}{\kappa_p(X)} : X>0\Big{\rbrace} = \inf\Big{\lbrace} \frac{\Em_{p}^\liou(X)}{\kappa_p(X)} : X>0, \| X\|_{1,\sigma} = 1\Big{\rbrace} 
\label{equ:beta}
\end{equation}
as replacing $X\mapsto X/\|X\|_{1,\sigma}$ does not change the value of the quotient $\Em_{p}^\liou(X)/\kappa_p(X)$. Therefore, to compute $\beta_p$ it is enough to optimize over relative densities (i.e.\ $X>0$ fulfilling $\| X\|_{1,\sigma} = 1$). By inserting $\beta_p$ into \eqref{equ:quotient} we have  
\[
\frac{d}{dt}D_p(e^{t\liou}(\rho)\|\sigma)\leq-2\beta_p(\liou)D_p\lb e^{t\liou}(\rho)\|\sigma\rb
\]
for any $\rho\in\D_d$ and Liouvillian $\liou:\M_d\ra\M_d$ with full rank fixed point $\sigma\in\D^+_d$. By integrating this differential inequality we get

\begin{thm}\label{thm:betapexpdecay}
Let $\liou:\M_d\ra\M_d$ be a Liouvillian with full rank fixed point $\sigma\in\D^+_d$. For any $p\geq1$ and $\rho\in\D_d$ we have
\begin{equation}
D_p\lb e^{t\liou}(\rho) \|\sigma\rb\leq e^{-2\beta_p(\liou) t}D_{p}\lb\rho\|\sigma\rb 
\end{equation}
where $\beta_p(\liou)$ is the constant defined in \eqref{equ:betaSup}.
\end{thm}

\subsection{Computing \texorpdfstring{$\beta_p$}{bp} in simple cases}\label{sec:simplecases}

In general it is not clear how to compute $\beta_p$ and it does not depend on spectral data of $\liou$ alone. This is not surprising, as the computation of the usual logarithmic Sobolev constants $\alpha_2$ or $\alpha_1$ is also challenging and the exact values are only known for few Liouvillians~\cite{diaconis1996,Kastoryanosob,relentdep}. 
In the following we compute $\beta_2$ for the depolarizing semigroups.

\begin{thm}[$\beta_2$ for the depolarizing Liouvillian]
Let $\liou_\sigma:\M_d\ra\M_d$ denote the depolarizing Liouvillian given by $\liou_\sigma(\rho) = \text{tr}\lb\rho\rb\sigma - \rho$ with fixed point $\sigma\in\D^+_d$. Then
\begin{equation}
\beta_2(\liou_\sigma)=\frac{1-\frac{1}{\|\sigma^{-1}\|_\infty}}{\ln\lb\|\sigma^{-1}\|_\infty\rb}.  
\end{equation}
\label{thm:betaForDep}
\end{thm}
\begin{proof}
Without loss of generality we can restrict to $X>0$ with $\|X\|_{1,\sigma}=1$ in the minimization \eqref{equ:beta}. 
Observe that the generator of the time evolution of the relative density (see \eqref{equ:HutLiou}) for the depolarizing Liouvillian is
\[
\hat{\liou}_\sigma(X)=\tr\lb\sigma^{\frac{1}{2}}X\sigma^{\frac{1}{2}}\rb\one-X .
\]
An easy computation yields $\Em^{\liou_\sigma}_2(X)=\|X\|_{2,\sigma}^2-1$ and so
\[
 \frac{\Em_2^{\liou_\sigma}(X)}{\kappa_2(X)}=\frac{1-\frac{1}{\| X\|_{2,\sigma}^2}}{\ln\lb\| X\|_{2,\sigma}^2\rb}.
\]
As the function $x\mapsto\frac{1-\frac{1}{x}}{\ln\lb x\rb}$ is monotone decreasing for $x\geq1$, we have
\begin{equation}
 \inf\limits_{X>0}\frac{\Em_2^{\liou_\sigma}(X)}{\kappa_2(X)}=\frac{1-\frac{1}{\|\sigma^{-1}\|_\infty}}{\ln\lb\|\sigma^{-1}\|_\infty\rb}, 
\end{equation}
where we used 
\[
\sup\limits_{X\geq0,\|X\|_{1,\sigma}=1}\| X\|_{2,\sigma}^2=\|\sigma^{-1}\|_\infty, 
\]
which easily follows from Lemma \ref{lem:maximalrelent} by exponentiating both sides of Equation \eqref{equ:upperentropy} and using the correspondence between
relative densities and states.
\end{proof}
The exact value of $\alpha_2(\liou_\sigma)$ is open to the best of our knowledge, but in the case of $\sigma=\frac{\one}{d}$ we have
$\alpha_2\lb\liou_\frac{\one}{d}\rb=\frac{2(1-2/d)}{\ln(d-1)}$~\cite[Theorem 24]{Kastoryanosob}, which is of the same order of magnitude as $\beta_2$ for these semigroups.

Computing $\beta_p$ for $p\not=2$ seems not to be straightforward even for depolarizing channels, but for the semigroup depolarizing to the maximally mixed state we can at 
least provide upper and lower bounds.

\begin{thm}[$\beta_p$ for the Liouvillian depolarizing to the maximally mixed state]\label{thm:boundbetapdep}
 Let $\liou(\rho)=\tr(\rho)\frac{\one}{d}-\rho$. For $p\geq2$ we have
\[
\frac{p}{2(p-1)}\frac{1}{\ln(d)}\geq\beta_p(\liou)\geq\frac{p}{2(p-1)}\frac{d^{\frac{p-1}{p}}-1}{d^{\frac{p-1}{p}}\ln(d)}.  
\]

\end{thm}
\begin{proof}
The Dirichlet Form of this Liouvillian for $X>0$ with $\|X\|_{1,\frac{\one}{d}}=1$ is given by
\[
\Em^\liou_p(X)=\frac{p}{2(p-1)}(\|X\|^p_{p,\frac{\one}{d}}-\|X\|^{p-1}_{p-1,\frac{\one}{d}}). 
\]
Dividing this expression by $\kappa_p(X)$ we get
\begin{equation}\label{equ:dirioverkappadep}
\frac{\Em^\liou_p(X)}{\kappa_p(X)}=\frac{1-\frac{\|X\|^{p-1}_{p-1,\frac{\one}{d}}}{\|X\|^p_{p,\frac{\one}{d}}}}{2\ln\lb\|X\|_{p,\frac{\one}{d}}\rb} .
\end{equation}
By the monotonicity of the weighted norms, we have 
\[
\frac{\|X\|^{p-1}_{p-1,\frac{\one}{d}}}{\|X\|^p_{p,\frac{\one}{d}}}\leq\frac{1}{\|X\|_{p,\frac{\one}{d}}}
\]
and so
\begin{align}\label{equ:lowerpdep}
\frac{\Em^\liou_p(X)}{\kappa_p(X)}\geq \frac{\|X\|_{p,\frac{\one}{d}}-1}{2\|X\|_{p,\frac{\one}{d}}\ln\lb\|X\|_{p,\frac{\one}{d}}\rb} 
\end{align}
The expression on the right-hand side of \eqref{equ:lowerpdep} is monotone decreasing in $\|X\|_{p,\frac{\one}{d}}$ and so the infimum is attained at
\[
\sup\limits_{\|X\|_{1,\frac{\one}{d}}=1} \|X\|_{p,\frac{\one}{d}}=d^{\frac{p-1}{p}},
\]
which again easily follows from Lemma \ref{lem:maximalrelent}.
The upper bound follows from \eqref{equ:dirioverkappadep} as
\[
\frac{\Em^\liou_p(X)}{\kappa_p(X)}\leq \frac{1}{2\ln\lb\|X\|_{p,\frac{\one}{d}}\rb}.
\]  
which is again monotone decreasing in $\|X\|_{p,\frac{\one}{d}}$.
\end{proof}
From the relations between LS constants~\cite[Proposition 13]{Kastoryanosob}, it follows that for the LS constants of the depolarizing channels we have
$\alpha_p\lb\liou_\frac{\one}{d}\rb\geq\alpha_2\lb\liou_\frac{\one}{d}\rb=\frac{2(1-2/d)}{\ln(d-1)}$ for $p\geq1$. The constants $\beta_p$ and $\alpha_p$
are therefore of the same order in this case for small $p\geq2$.

\section{Comparison with similar quantities}

\subsection{Comparison with spectral gap}

Here we show how $\beta_p$, see \eqref{equ:betaSup}, compares to the spectral gap \eqref{equ:SpectralGap} of a Liouvillian.
This is motivated by similar results for logarithmic Sobolev constants, where it was shown~\cite[Theorem 16]{Kastoryanosob} that
$\alpha_1(\liou)\leq\lambda(\liou)$ for reversible semigroups, a result we recover and generalize here.
\begin{thm}[Upper bound spectral gap]\label{thm:upperboundpspectral}
Let $\liou:\M_d\ra\M_d$ be a primitive and reversible Liouvillian with full rank fixed point $\sigma\in\D^+_d$ and $p\geq1$. Then 
\begin{equation}
\beta_p\lb\liou\rb\leq\lambda\lb\liou\rb.
\label{equ:CompBetapSpec}
\end{equation}

\end{thm}
\begin{proof}
Let $(s_i)^d_{i=1}$ denote the spectrum of $\sigma^{1/p}$ and choose a unitary $U$ such that 
\[
\sigma^{1/p} = U\text{diag}\lb s_1, s_2, \ldots , s_d\rb U^\dagger.
\]
As $\liou$ is reversible, there is a self-adjoint eigenvector
$X\in\M_d$ of $\hat{\liou}$ corresponding to the spectral gap, i.e. $\hat{\liou}(X) = -\lambda(\liou) X$. Let $\epsilon_0>0$ be small enough such that $Y_\epsilon = \one_d + \epsilon X$ is positive for any $|\epsilon|\leq \epsilon_0$. 
For $|\epsilon|\leq \epsilon_0$ we use Lemma \ref{lem:TaylorExp} of the appendix to show
\begin{equation}
\beta_p(\liou) \leq \frac{\Em^\liou_p(Y_\epsilon)}{\kappa_p(Y_\epsilon)} 
= \frac{\lambda(\liou)\frac{p}{2(p-1)}\lb 2\epsilon^2\sum_{1\leq i\leq j\leq d}f_p(s_i,s_j)b_{ij}b_{ji} + O(\epsilon^3)\rb}
{\frac{\epsilon^2}{p-1}\lb p\sum_{1\leq i\leq j\leq d}f_p(s_i,s_j)b_{ij}b_{ji}\rb + O(\epsilon^3)}
\label{equ:LongQuotient}
\end{equation}
where $b_{ij} = (U^\dagger \sigma^{1/{2p}}X\sigma^{1/{2p}}U)_{ij}$ and 
\begin{equation}
f_p(x,y) = \begin{cases} (p-1)x^{p-2} & \text{ if }x=y \\
\frac{x^{p-1}-y^{p-1}}{x-y} & \text{ else.}\end{cases}
\end{equation} 
Observe that $f_p(s_i,s_j)>0$ for $s_i,s_j>0$. Moreover, as  $U^\dagger \sigma^{1/{2p}}X\sigma^{1/{2p}}U$ is non-zero and self-adjoint we have $b_{ij}b_{ji}\geq 0$ for all $i,j$ and this inequality is strict for at least one choice of $i,j$. Therefore, the terms of second order in $\epsilon$ in the numerator and denominator of \eqref{equ:LongQuotient} are strictly positive, and we obtain $\lambda(\liou)$ as the limit of the quotient as $\epsilon\to 0$.

\end{proof}

A similar argument as the one given in the previous proof shows that all real, nonzero elements of the spectrum of $\hat{\liou}$ are upper bounds to $\beta_p$ without invoking reversibility. 

Note that in the case of $p=2$ (see the discussion after \eqref{equ:LogSob}) we may assume that the Liouvillian is reversible without loss of generality and
drop the requirement of reversibility in the previous theorem. 
Alternatively, we can obtain the same statement directly from a simple functional inequality. In this case we can also give a lower bound on $\beta_2$ in terms of the spectral gap.

\begin{thm}[Upper and lower bound for $\beta_2$]
Let $\liou:\M_d\ra\M_d$ be a primitive Liouvillian with full rank fixed point $\sigma\in\D^+_d$. Then
\begin{equation}
\lambda\lb\liou\rb\frac{1-\frac{1}{\|\sigma^{-1}\|_\infty}}{\ln\lb \|\sigma^{-1}\|_\infty\rb}  \leq\beta_2\lb\liou\rb\leq\lambda\lb\liou\rb .
\label{equ:CompBetaSpectoBeta2}
\end{equation}
\label{thm:CompBetaSpec} 
\end{thm}

To prove Theorem \ref{thm:CompBetaSpec} we need the following Lemma.

\begin{lem}
For any $X\in\M_d$ we have 
\[
\Var_\sigma(X)\leq\kappa_2(X).
\]
\label{lem:VarVsKappa}
\end{lem}
\begin{proof}
For $X>0$ dividing both sides of the inequality by $\| X\|^2_{1,\sigma}$ yields
\[
\frac{\| X\|^2_{2,\sigma}}{\| X\|^2_{1,\sigma}} - 1\leq \frac{\| X\|^2_{2,\sigma}}{\| X\|^2_{1,\sigma}}\ln\lb\frac{\| X\|^2_{2,\sigma}}{\| X\|^2_{1,\sigma}}\rb.
\] 
This follows from the elementary inequality $x-1\leq x\ln(x)$ for $x\geq 1$, where we use the ordering $\|X\|_{2,\sigma}\geq \| X\|_{1,\sigma}$ for any $X\in\M_d$. 
 
\end{proof}

\begin{proof}[Proof of Theorem \ref{thm:CompBetaSpec}]
Using the definition of $\beta_2$ (see \eqref{equ:betaSup}) and Lemma \ref{lem:VarVsKappa} yields
\[
 \beta_2\Var_{\sigma}(X)\leq \beta_2\kappa_2(X)\leq\Em^\liou_2(X).
\]
Now the variational definition of $\lambda(\liou)$ (see \eqref{equ:SpectralGap}) implies the second inequality of \eqref{equ:CompBetaSpectoBeta2}.

To prove the first inequality of \eqref{equ:CompBetaSpectoBeta2} consider the depolarizing Liouvillian 
\[
\liou_{\sigma}(X)=\tr(X)\sigma-X .
\]
By Theorem \ref{thm:betaForDep} we have  
 \[
  \frac{1-\frac{1}{\|\sigma^{-1}\|_\infty}}{\ln\lb \|\sigma^{-1}\|_\infty\rb} \kappa_2(X)\leq\Em_2^{\liou_\sigma}(X) 
 \]
As $\Em_2^{\liou_\sigma}(X)=\Var_{\sigma}(X)$, we have $\Em_2^{\liou_\sigma}(X)\leq \frac{1}{\lambda(\liou)}\Em^{\liou}_2(X)$ by the variational definition of $\lambda(\liou)$ (see \eqref{equ:SpectralGap}). 
Inserting this in the above inequality finishes the proof.

\end{proof}

\subsection{Comparison with logarithmic Sobolev constants}

Here we show how $\beta_p$, see \eqref{equ:betaSup}, compares to the logarithmic Sobolev constant $\alpha_p$.

\begin{thm}\label{thm:betaVsAlpha}
Let $\liou:\M_d\to\M_d$ be a primitive Liouvillian with full rank fixed point $\sigma\in\D_d^+$. Then for any $p \geq 1$ we have
\begin{equation}
\beta_p\lb\liou\rb\geq\frac{\alpha_p\lb\liou\rb }{p}.
\label{equ:betaVsAlpha}
\end{equation}

\end{thm}

We will need the following Lemma.

\begin{lem}
For any full rank state $\sigma\in\D_d^+$, any $p>1$ and $X\in\M^+_d$ with $\| X\|_{1,\sigma}=1$ we have
\begin{align}
\Ent_{p,\sigma}(X)\geq\frac{\kappa_{p}(X)}{p}.  
\end{align}

\label{lem:EntVsKappa}
\end{lem}
\begin{proof}
The function $p\mapsto D_p\lb\rho\|\sigma\rb$ is monotonically increasing~\cite{muller2013quantum,dataprocessingbeigi} and differentiable (as the noncommutative $l_p$-norm is differentiable in $p$~\cite[Theorem 2.7]{Olkiewicz1999246}). 
Thus, with $f:\R^+\ra \R$ given by $f(t) = t+p$ we have
\[
0\leq \|X\|^p_{p,\sigma}\frac{d}{dt} \lb D_{f(t)}\lb\rho\|\sigma\rb\rb\Big{|}_{t=0} = -\frac{1}{(p-1)^2}\|X\|^p_{p,\sigma}\ln\lb \| X\|^p_{p,\sigma}\rb + \frac{1}{p-1}\frac{d}{dt}\lb \|X\|^{f(t)}_{f(t),\sigma}\rb\Big{|}_{t=0}.
\]  
where we used the relative density $X = \Gamma^{-1}_\sigma\lb\rho\rb$. The remaining derivative in the above equation has been computed in \cite[Theorem 2.7]{Olkiewicz1999246} and we have 
\[
\frac{d}{dt}\lb \|X\|^{f(t)}_{f(t),\sigma}\rb\Big{|}_{t=0} = \lk I_{q,p}(X),S_p(X)\rk_\sigma
\]
with the operator valued entropy $S_p$ defined in \eqref{equ:OpValEntr} and $\frac{1}{p}+\frac{1}{q}=1$. Inserting this expression in the above equation we obtain
\begin{equation}
\frac{1}{p-1}\|X\|^p_{p,\sigma}\ln\lb \| X\|^p_{p,\sigma}\rb \leq \lk I_{q,p}(X),S_p(X)\rk_\sigma .
\label{equ:Krams}
\end{equation}
for any $X\in\M^+_d$ with $\|X\|_{1,\sigma}=1$, i.e. for any $X = \Gamma^{-1}_\sigma\lb\rho\rb$ for some state $\rho\in\D_d$. Now we get
\begin{align*}
p\Ent_{p,\sigma}(X) &= p\lk I_{q,p}(X),S_p(X)\rk_\sigma - \|X\|^p_{p,\sigma}\ln(\|X\|^p_{p,\sigma}) \\
&\geq \frac{p}{p-1}\|X\|^p_{p,\sigma}\ln(\| X\|^p_{p,\sigma}) - \|X\|^p_{p,\sigma}\ln(\|X\|^p_{p,\sigma}) \\
& = \kappa_p(X)
\end{align*}
where we used \eqref{equ:Krams}.

\end{proof}

\begin{proof}[Proof of Theorem \ref{thm:betaVsAlpha}]
There is nothing to show for $p=1$ as $\alpha_1(\liou)=\beta_1(\liou)$ and we can assume $p>1$. For $X\in\M^+_d$ with $\| X\|_{1,\sigma}=1$ we can use Lemma \ref{lem:EntVsKappa} and the definition of $\alpha_p\lb\liou\rb$ to compute
\[
\frac{\alpha_p\lb \liou\rb}{p}\kappa_p(X)\leq\alpha_p\lb\liou\rb\Ent_{p,\sigma}(X)\leq\Em^\liou_p(X). 
\]
By the variational definition \eqref{equ:betaSup} of $\beta_p$ the claim follows.
\end{proof}

Theorem \ref{thm:betaVsAlpha} will be applied in Section \ref{sec:mixingtimes} to obtain bounds on the mixing time of a Liouvillian with a positive logarithmic Sobolev constant without invoking any form of $l_p$-regularity 
(see \cite{Kastoryanosob}).
As usually a logarithmic Sobolev is implied by a hypercontractive inequality~\cite{Olkiewicz1999246}, we would like to remark that
one can also make a similar statement as that of Theorem \ref{thm:betaVsAlpha} from a hypercontractive inequality.
One can easily show that 
\begin{align}
||e^{t\hat{\liou}}||_{p(t)\to p,\sigma}\leq1  
\end{align}
for $p(t)=(p-1)e^{-\alpha_p t}+1$ implies that $\beta_p(\liou)\ge\frac{\alpha_p}{p}$.

%
%
%

\section{Mixing times}\label{sec:mixingtimes}

In this section we will introduce the quantities of interest and prove the building blocks to prove mixing times from the entropy production
inequalities of the last sections, distinguishing between continuous and discrete time semigroups.
We will mostly focus on $\beta_2$, as this seems to be the most relevant constant for mixing time applications.
This is justified by the fact that the underlying Dirichlet form is a quadratic form and the entropy related to it
stems from a Hilbert space norm. Moreover, as the same Dirichlet form is also involved in computations of the spectral 
gap, it could be easier to adapt existing techniques, such as the ones developed in \cite{lowerbounddavies,Temme2017}.
\begin{defn}[Mixing times]
For either $I=\R^+$ or $I = \N$ let $\{T_t\}_{t\in I}$ be a primitive semigroup of quantum channels with fixed point $\sigma\in\D_d^+$. We define the $l_1$ \textbf{mixing time} for $\epsilon>0$ as
\[
 t_1(\epsilon)=\inf\{t\in I~:~\|T_t(\rho)-\sigma\|_{1}\leq\epsilon\text{ for all }\rho\in\D_d \}.
\]
Similarly we define the $l_2$ \textbf{mixing time} for $\epsilon>0$ as
\[
t_2(\epsilon)=\inf\{t\in I~:~\Var_\sigma\lb\hat{T_t}\lb X\rb\rb\leq\epsilon\text{ for all }X\in\M^+_d \text{ with }\|X\|_{1,\sigma}=1\}. 
\]
In the continuous case $I=\R^+$ we will often speak of the mixing times of the Liouvillian generator of a quantum dynamical semigroup which we identify with the mixing times of the semigroup according to the above definition.

\end{defn}

\subsection{Mixing in Continuous Time}
It is now straightforward to get mixing times from the previous results.

\begin{thm}[Mixing time from entropy production]\label{thm:l1mixingfrombp}
Let $\liou:\M_d\to\M_d$ be a primitive Liouvillian with fixed point $\sigma\in\D^+_d$. Then 
\[
t_1(\epsilon)\leq\frac{1}{2\beta_p(\liou)}\ln\lb\frac{2\ln\lb \|\sigma^{-1}\|_\infty\rb}{\epsilon^2}\rb.
\]

\end{thm}
\begin{proof}
From \eqref{equ:pinskerrenyi} and Lemma \ref{lem:maximalrelent} we have
\begin{align}
\ln\lb \|\sigma^{-1}\|_\infty\rb e^{-2\beta_p(\liou) t}\geq\frac{1}{2}\|e^{t\liou}(\rho)-\sigma\|_1^2. 
\end{align}
for any $\rho\in\D_d$. The claim follows after rearranging the terms.
\end{proof}

Using Theorem \ref{thm:betaVsAlpha} we can lower bound $\beta_p$ in terms of the usual logarithmic Sobolev constant $\alpha_p$. Combining this with Theorem \ref{thm:l1mixingfrombp} shows the following Corollary.  

\begin{cor}[Mixing time bound from logarithmic Sobolev inequalities]
Let $\liou:\M_d\ra\M_d$ be a primitive Liouvillian with fixed point $\sigma\in\D^+_d$. Then 
\begin{align}\label{equ:boundmixing}
t_1(\epsilon)\leq\frac{p}{2\alpha_p(\liou)}\ln\lb\frac{2\ln\lb \|\sigma^{-1}\|_\infty\rb}{\epsilon^2}\rb.
\end{align}

\label{cor:MixingFromLogSob}
\end{cor}
By Corollary \ref{cor:MixingFromLogSob} a nonzero logarithmic Sobolev constant always implies a nontrivial mixing time bound.
One should say that the same bound was showed in~\cite{Kastoryanosob} for $p=2$, however under additional assumptions 
(specifically $l_p$-regularity~\cite{Kastoryanosob}) on the Liouvillian in question. While these assumptions have been shown for certain classes of Liouvillians 
(including important examples like Davies generators and doubly stochastic Liouvillians~\cite{Kastoryanosob}) they have not been shown in general.   
Moreover, the bound in Theorem \ref{thm:l1mixingfrombp} clearly does not depend on $p$ and one could in principle optimize over all $\beta_p$. However, as the computations in subsection \ref{sec:simplecases} already indicate,
it does not seem to be feasible to compute or bound $\beta_p$ for $p\not=2$ even in simple cases and one will probably only work with $\beta_2$ in applications.

The bound from Corollary \ref{cor:MixingFromLogSob} also has the right scaling properties needed in recent applications of rapid mixing, such as the results in \cite{stabcubitt,rapidarealaw}. In particular, together with the results in 
\cite{quasifreesob}, the last Corollary shows that the hypothesis of Theorem 4.2 in \cite{stabcubitt} is always satisfied 
for product evolutions and not only for the special classes considered in \cite{Kastoryanosob}.  
%
%
%
%

One may also use these techniques to get mixing times in the $l_2$ norms which are stronger than the ones obtained just by considering that $\beta_2$ is a lower bound to the spectral gap.

\begin{thm}[$l_2$-mixing time bound]
Let $\liou:\M_d\to\M_d$ be a Liouvillian with fixed point $\sigma\in\D_d^+$. Then 
\begin{equation}\label{eq:t2frombeta2}
t_2(\epsilon)\leq\frac{1}{2\beta_2(\liou)}\ln\lb\frac{\ln\lb \|\sigma^{-1}\|_\infty\rb}{\ln(1+\epsilon)}\rb.
\end{equation}

\end{thm}

\begin{proof}
For $X>0$ with $\|X\|_{1,\sigma}=1$ we have $\Var_\sigma(X)=\| X-\one\|_{2,\sigma}^2=\| X\|_{2,\sigma}^2-1$ and thus
\[
\kappa_2(X)=\lb 1+\Var_\sigma(X)\rb\ln\lb 1+\Var_\sigma(X)\rb. 
\]
In the following let $X_t = e^{t\hat{\liou}}\lb X\rb$ denote the time evolution of the relative density $X$. Using \eqref{equ:DerivVar} and the definition of $\beta_2(\liou)$ (see \eqref{equ:betaSup}) we obtain
\[
\frac{d}{dt}\Var_\sigma(X_t)=-2\Em_2^\liou(X_t) \leq -2\beta_2(\liou)(1+\Var_\sigma(X_t))\ln(1+\Var_\sigma(X_t)).
\]
Integrating this differential inequality we obtain
\begin{align*}
\ln\lb\frac{\ln(1+\Var_\sigma(X))}{\ln(1+\epsilon)}\rb &\leq \int\limits_{0}^{t_2(\epsilon)}\frac{1}{\lb 1+\Var_\sigma(X_t)\rb\ln\lb 1+\Var_\sigma(X_t)\rb}\lbr\frac{d}{dt}\Var_\sigma(X_t)\rbr dt \\
&\leq-2\beta_2t_2(\epsilon) .
\end{align*}
As $1+\Var_\sigma(X)\leq\|\sigma^{-1}\|_\infty$, the claim follows after rearranging
the terms.

\end{proof}

%
%
%
%
%
%

In the remaining part of the section we will discuss a converse to the previous mixing time bounds, i.e.\ a lower bound on the logarithmic Sobolev constant in terms of a mixing time. This excludes the possibility of a reversible semigroup with both small $\beta_2$ and short mixing time with respect to the $l_2$ distance. For this we generalize \cite[Corollary 3.11]{diaconis1996}
to the noncommutative setting. 

\begin{thm}[LS inequality from $l_2$ mixing time]\label{thm:uncertaintytau2ls2}
 Let $\liou:\M_d\to\M_d$ be a primitive, reversible Liouvillian with
 fixed point $\sigma\in\D^+_d$. Then
 \begin{equation}\label{eq:uncertaintytau2ls2}
  \frac{1}{2}\leq\alpha_2\lb \liou\rb t_2\lb e^{-1}\rb \leq 2\beta_2\lb \liou\rb t_2\lb e^{-1}\rb. 
 \end{equation}
Moreover, this inequality is tight.
\end{thm}
\begin{proof}
 We refer to Appendix \ref{sec:l2uncertproof} for a proof.
\end{proof}

As remarked in \cite{diaconis1996}, even the classical result does not hold anymore if we drop the reversibility assumption. Therefore, this assumption is also needed in the noncommutative setting. By considering a completely depolarizing channel it is also easy to see that no such bound can hold in discrete time.

%
%

Theorem \ref{thm:uncertaintytau2ls2} implies that for reversible Liouvillians $\beta_2$ and $\alpha_2$ cannot differ by a large factor. More specifically we have the following corollary.

\begin{cor}\label{cor:lowersobbeta2}
 Let $\liou:\M_d\to\M_d$ be a primitive, reversible Liouvillian with
 fixed point $\sigma\in\D^+_d$. Then
\begin{equation}
2\beta_2(\liou)\geq \alpha_2(\liou)\geq\beta_2(\liou)\ln\lb\frac{\ln\lb \|\sigma^{-1}\|_\infty\rb}{\ln(1+e^{-1})}\rb.
\end{equation}

\end{cor}
\begin{proof}
We showed the first inequality in Theorem \ref{thm:betaVsAlpha}. The second inequality follows by combining \eqref{eq:uncertaintytau2ls2} and \eqref{eq:t2frombeta2}.
\end{proof}

\subsection{Mixing in Discrete Time}

In this section we will obtain mixing time bounds and also entropic inequalities
for discrete-time quantum channels $T:\M_d\to\M_d$. 
We will then use these techniques to derive mixing times for random local channels, which we will define next.
These include channels that usually appear in quantum error correction scenarios, such as random Pauli errors on qubits~\cite[Chapter 10]{nielsen2000quantum}.
They will be based on the following quantity:

\begin{defn}
For a primitive quantum channel $T:\M_d\to\M_d$ with full rank fixed point $\sigma\in\D_d^+$, we define
\begin{align}
\beta_D(T)=\beta_2(T^*\hat{T}-\id_d). 
\end{align}
Here we used $\hat{T}=\Gamma_\sigma^{-1}\circ T\circ\Gamma_\sigma$.

\end{defn}

The definition of $\beta_D(T)$ can be motivated by the following improved data-processing inequality for the $2$-sandwiched R\'{e}nyi divergence.

\begin{thm}\label{thm:improveddatadiscrete}
Let $T:\M_d\to\M_d$ be a primitive quantum channel with full rank fixed point $\sigma\in\D^+_d$. Then for all $\rho\in\D_d$ we have
\begin{align}\label{equ:improveddatadiscrete}
D_2\lb T(\rho)\|\sigma\rb\leq(1-\beta_D(T))D_2\lb\rho\|\sigma\rb. 
\end{align}

\end{thm}
\begin{proof}
Let $X=\sigma^{-\frac{1}{2}}\rho\sigma^{-\frac{1}{2}}$ denote the relative density of $\rho$ with respect to $\sigma$. Observe that the 2-Dirichlet form (see Definition \ref{defn:DirichletForm}) of the semigroup $\liou=T^*\hat{T}-\id_d$ can be written as
\[
\Em_2^\liou(X)=\|X\|_{2,\sigma}^2-\|\hat{T}(X)\|_{2,\sigma}^2.  
\]
From the definition of $\beta_2(\liou)$ (see \eqref{equ:betaSup}) it follows that
\[
\|X\|_{2,\sigma}^2-\|\hat{T}(X)\|_{2,\sigma}^2\geq\beta_2\kappa_2(X), 
\]
which is equivalent to
\[
\ln(\|\hat{T}(X)\|_{2,\sigma}^2)-\ln(\|X\|_{2,\sigma}^2)\leq\ln(1-\beta_2\ln(\|X\|_{2,\sigma}^2).
\]
Using the elementary inequality $\ln(1-\beta_2\ln(\|X\|_{2,\sigma}^2)\leq-\beta_2\ln(\|X\|_{2,\sigma}^2)$, that $\ln(\|\hat{T}(X)\|_{2,\sigma}^2)=D_2(T(\rho)\|\sigma)$ and $\ln(\|X\|_{2,\sigma}^2)=D_2(\rho\|\sigma)$ hold, the statement of the theorem follows after rearranging the terms.
\end{proof}

One should note that, unlike in Theorem \ref{thm:betapexpdecay}, the constant $\beta_D$ is \emph{not} optimal in \eqref{equ:improveddatadiscrete}. As an example take $T(\rho)=\tr[\rho]\frac{\one_d}{d}$ for which $\beta_D(T)=\frac{1-d^{-1}}{\ln\lb d\rb}$, but $D_2\lb T(\rho)\|\frac{\one_d}{d}\rb = 0$. 
Also, $\beta_D(T)>0$ is not a necessary condition for primitivity, as there are primitive quantum channels that are not strict contractions with respect to $D_2$. To see this, consider the map $T:\M_2\to\M_2$ which acts as 
follows on Pauli operators:
\[
T(\one)=\one, \quad T(\sigma_x)=0,\quad T(\sigma_y)=0 \text{ and}\quad T(\sigma_z)=\sigma_x.
\]
One can check that this is a a primitive quantum channel with $T^2(\rho)=\frac{\one}{2}$ for any state $\rho\in\D_d$. However, $T$ maps the pure state $\frac{1}{2}(\one+\sigma_z)$ to the pure state
$\frac{1}{2}(\one+\sigma_x)$, which implies that $D_2$ does not strictly contract under $T$. 
We can now prove the following bound on the discrete mixing time.

\begin{thm}[Discrete mixing time]\label{thm:mixingdiscrete}
Let $T:\M_d\to\M_d$ be a primitive quantum channel with full rank fixed point $\sigma\in\D_d^+$ and $\beta_D(T)>0$. Then
\[
t_1(\epsilon)\leq-\frac{1}{\ln(1-\beta_D(T))}\ln\lb\frac{2\ln\lb \|\sigma^{-1}\|_\infty\rb}{\epsilon^2}\rb. 
\]

\end{thm}
\begin{proof}
By Theorem \ref{thm:improveddatadiscrete} we have 
\[
D_2(T^n(\rho)\|\sigma)\leq(1-\beta_D(T))^n D_2(\rho\|\sigma). 
\]
for any $\rho\in\D_d$. The claim then follows from \eqref{equ:pinskerrenyi} and Lemma \ref{lem:maximalrelent}.
\end{proof}

Convergence results for primitive continuous-time semigroups can often be lifted to their tensor powers. In discrete time a similar result holds for the following class of channels:

\begin{defn}[Random Local Channels]
For a quantum channel $T:\M_d\to\M_d$ and probabilities $\mathbf{p}=(p_1,\ldots ,p_n)$ with $p_i\geq 0$ and $\sum_i p_i=1$ we define a \textbf{random local channel} $T_{\mathbf{p}}^{(n)}:\M_d^{\otimes n}\to\M_d^{\otimes n}$ by
\begin{equation}\label{equ:channelconvex}
T^{(n)}_{\mathbf{p}} =\sum_{i=1}^n p_i~ \id_d^{\otimes i-1}\otimes T\otimes\id_d^{\otimes n-i}.
\end{equation}
\end{defn}

The previous definition can be generalized to the case where not all local channels are identical, i.e. if we have $T_i:\M_d\ra\M_d$ acting on the $i$th system in the expression \eqref{equ:channelconvex}. As long as the local channels are all primitive our results also hold for this more general class of channels. However, for simplicity we will restrict here to the above definition.

\begin{thm}\label{thm:tensorizdisc}
Let $T:\M_d\ra\M_d$ be a primitive quantum channel with full rank fixed point $\sigma\in\D_d^+$ such that the Liouvillian $\hat{\liou}=T^*\hat{T}-\id_d$ fulfills $\beta_2\lb\liou^{(n)}\rb\geq q$ for some $q>0$ and all $n\in\N$. Then for any $n\in\N$ and probabilities $\mathbf{p}=(p_1,\ldots ,p_n)$ with $p_i\geq 0$ and $\sum_i p_i=1$ we have
\begin{equation}
D_2(T^{(n)}_{\mathbf{p}}(\rho)\|{\sigma^{\otimes n}})\leq(1-qp_{\min}^2)D_2(\rho\|{\sigma^{\otimes n}}), 
\end{equation}
for any $\rho\in\D_{d^n}$ and where $p_{\min} = \min p_i$.
\end{thm}
\begin{proof}
By Theorem \ref{thm:improveddatadiscrete} it is enough to show that $\beta_D(T_{\mathbf{p}})\geq qp_{\min}^2$. 

Observe that the Dirichlet form of $(T^{(n)}_{\mathbf{p}})^*\hat{T}^{(n)}_{\mathbf{p}}-\id_d$ is given by
\[
\Em_2^\liou(X)=\sum_{i\not=j}p_ip_j\lk X-T_i^*\hat{T}_{j}(X),X\rk_{\sigma^{\otimes n}}+\sum_ip_i^2\lk X-T^*_i\hat{T}_{i}(X),X\rk_{\sigma^{\otimes n}} 
\]
where the map $T_{i}^*\hat{T}_{j}$ acts as $T^*$ on the $i$-th system, $\hat{T}$ on the $j$-th and as the identity elsewhere. As $T_{i}^*\hat{T}_{j}\leq\id_d$ with respect to $\lk\cdot,\cdot\rk_{\sigma^{\otimes n}}$ we have
\[
\Em_2^\liou(X)\geq\sum_ip_i^2\lk X-T^*_i\hat{T}_{i}(X),X\rk_{\sigma^{\otimes n}}\geq p_{\min}^2\Em_2^{\liou^{(n)}}(X).
\]
From the comparison inequality $\Em_2^\liou\geq p_{\min}^2\Em_2^{\liou^{(n)}}$ and the assumption $\beta_2\lb\liou^{(n)}\rb\geq q$ it then follows that
$\beta_D(\Phi)\geq qp_{\min}^2$.
\end{proof}

%

As an application we can bound the entropy production and the mixing time in a system of $n$ qubits affected (uniformly) by random Pauli errors. 
The time evolution of this system is given by the channel $T_n:\M_2^{\otimes n}\to\M_2^{\otimes n}$ given by
\begin{equation}\label{eq:defrandpauli}
T_n=\frac{1}{n}\sum_{i=1}^{n}\id_2^{\otimes i-1}\otimes T\otimes\id_2^{\otimes n-i}
\end{equation}
with $T(\rho)=\tr(\rho)\frac{\one}{2}$.

\begin{thm}
For $T_n$ defined as in equation \eqref{eq:defrandpauli} we have
\[
D_2\lb T_n(\rho)\|\frac{\one_{2^n}}{2^n}\rb\leq\lb1-\frac{1}{2n^2}\rb D_2\lb\rho\|\frac{\one_{2^n}}{2^n}\rb.
\]
for any $\rho\in\D_{2^n}$.
\end{thm}
\begin{proof}
From \cite{Kastoryanosob} it is known that 
\[
\alpha_2\lb\liou^{(n)}_{\frac{\one}{2}}\rb=1 .
\]
Now combining Theorem \ref{thm:betaVsAlpha} and Theorem \ref{thm:tensorizdisc} gives 
\[
\beta_D(T_n)\geq\frac{1}{2n^2}\alpha_2\lb\liou^{(n)}_{\frac{\one}{2}}\rb.
\]
\end{proof}

\begin{cor}
Let $T_n$ be defined as in \eqref{eq:defrandpauli}. Then we have
\begin{align}
t_{1}(\epsilon)\leq -\frac{1}{\ln\lb1-\frac{1}{2n^2}\rb}\ln\lb\frac{n}{\epsilon^2}\rb.
\end{align}

\end{cor}
\begin{proof}
This follows directly from the previous theorem and Theorem \ref{thm:mixingdiscrete}. 
\end{proof}

\section{Strong converse bounds for the classical capacity}\label{sec:boundscap}

When classical information is sent via a quantum channel, the classical capacity is the supremum of transmission rates such that the probability for a decoding error vanishes in the limit of infinite channel uses. In general it is not possible to retrieve the information perfectly 
when it is sent over a finite number of uses of the channel, and the probability for successful decoding will be smaller than $1$. Here we want to derive bounds on this probability for quantum dynamical semigroups. More specifically we are interested in strong converse bounds on the classical capacity. An upper bound on the capacity is called a strong converse bound if whenever a transmission rate exceeds the bound the probability of successful decoding goes to zero in the limit of infinite channel uses. 

We refer to \cite[Chapter 12]{nielsen2000quantum} for the exact definition of the classical capacity and to \cite{796385,796386,strongconversedep,strongconvrenyi,tomamichel2015strong} for more details on strong converses and strong converse bounds. 


In \cite{strongconvrenyi} the following quantity was used to study strong converses.

\begin{defn}[$p$-information radius]
Let $T:\M_d\to\M_d$ be a quantum channel. The $p$-\textbf{information radius} $T$ is defined as\footnote{The $\ln(2)$ factor is due to our different choice of normalization for the divergences.}
\[
K_p(T)=\frac{1}{\ln(2)}\min\limits_{\sigma\in\D_d}\max\limits_{\rho\in\D_d}D_p(T(\rho)\|\sigma) .
\]

\end{defn}

We will often refer to a $(m,n,p)$-coding scheme for classical communication using a quantum channel $T$. By this we mean a coding-scheme for the transmission of $m$ classical bits via $n$ uses of the channel $T$ for which the probability of successful decoding is $p$ (see again \cite[Chapter 12]{nielsen2000quantum} for an exact definition). 
The following theorem shown in \cite[Section 6]{strongconvrenyi} relates the information radius and the probability of successful decoding. 

\begin{thm}[Bound on the success probability in terms of information radius\label{thm:boundpsucc}]
Let $T:\M_d\to\M_d$ be a quantum channel, $n\in\N$ and $R\geq 0$. For any $(nR,n,p_{\text{succ}})$-coding scheme for classical communication via $T$ we have
\begin{equation}\label{uppersuccess}
p_{\text{succ}}\leq 2^{-n\lb\frac{p-1}{p}\rb\lb R-\frac{1}{n}K_p\lb T^{\otimes n}\rb\rb}. 
\end{equation}

\end{thm}

We will now apply the methods developed in the last sections to obtain strong converse bounds on the capacity of quantum dynamical semigroups.\hfill\\
\begin{thm}\label{thm:renyitocapacity}
Let $\liou:\M_d\to\M_d$ be a primitive Liouvillian with full rank fixed point $\sigma\in\D^+_d$ such that for some $p\in(1,\infty)$ there exists $c>0$ fulfilling $\beta_p(\liou^{(n)})\geq c$ for all $n\in\N$. Then for any $(nR,n,p_{\text{such}})$-coding scheme for classical communication via the quantum dynamical semigroup $T_t = e^{t\liou}$ we have
\[
p_{succ}\leq 2^{-n\lb\frac{p-1}{p}\rb\lb R-e^{-2ct}\log(\|\sigma^{-1}\|_{\infty})\rb} .
\]
\end{thm}
\begin{proof}
Using Theorem \ref{thm:betapexpdecay} and Lemma \ref{lem:maximalrelent} we have
\[
K_p(T^{\otimes n}_t)\leq \frac{1}{\ln(2)}\max\limits_{\rho\in\D_{d^n}}D_p(T_t^{\otimes n}(\rho)\|\sigma)\leq ne^{-2\beta_p\lb\liou^{(n)}\rb t}\log(\|\sigma^{-1}\|_{\infty}).
\]
Now Theorem \ref{thm:boundpsucc} together with the assumption $\beta_p(\liou^{(n)})\geq c$ finishes the proof. 
\end{proof}

Together with Theorem \ref{thm:betaVsAlpha} the previous theorem shows that a quantum memory can only reliably store classical information for small times when it is subject to noise described by a quantum dynamical semigroup with ``large'' logarithmic Sobolev constant, as we will see more
explicitly later in Section \ref{sec:boundssemi}. 
Moreover, we can use the results from \cite{quasifreesob} to give a universal lower bound to the decay of the capacity in terms of the spectral gap and the fixed point.

\begin{cor}\label{cor:decaycap}
Let $\liou:\M_d\to\M_d$ be a primitive Liouvillian with full rank fixed point $\sigma\in\D^+_d$ and spectral gap $\lambda(\liou)$. Then for any $(nR,n,p_{\text{such}})$-coding scheme for classical communication via the quantum dynamical semigroup $T_t = e^{t\liou}$ we have
\[
p_{\text{such}}\leq 2^{-\frac{n}{2}\lb R-e^{-k(\lambda,\sigma)t}\log(\|\sigma^{-1}\|_{\infty})\rb} 
\]
where  $k(\lambda,\sigma):=\frac{\lambda}{2\lb\ln\lb d^4 \|\sigma^{-1}\|_{\infty}\rb+11\rb}$.
\end{cor}
\begin{proof}
It was shown in \cite[Theorem 9]{quasifreesob} that $\alpha_2(\liou^{(n)})\geq 2k(\lambda(\liou),\sigma)$ for all $n\in\N$. Using Theorem \ref{thm:betaVsAlpha} we have $\beta_2(\liou^{(n)})\geq k(\lambda(\liou),\sigma)$ for all $n\in\N$. Together with Theorem \ref{thm:renyitocapacity} this gives the claim.
\end{proof}

For unital semigroups, i.e. for $\sigma=\frac{\one_d}{d}$, one can improve the bound from the previous theorem slightly using (see \cite[Theorem 3.3]{entproddoubly}) 
\begin{equation}\label{equ:boundlsunital}
k\lb\lambda,\frac{\one_d}{d}\rb=\frac{\lambda(1-2d^{-2})}{2(\ln(3)\ln(d^2-1)+2(1-2d^{-2})}
\end{equation}
For $d=2$ we even have $k(\lambda,\frac{\one_d}{2})=\frac{\lambda}{2}$ (see \cite{Kastoryanosob}). 

\section{Examples of bounds for the classical capacity of Semigroups}\label{sec:boundssemi}
We will now apply the estimate on the capacity given by Corollary \ref{cor:decaycap} to some examples of semigroups. Here $C(T)$ will denote the classical
capacity of a quantum channel $T$.

\subsection{Depolarizing Channels}
In \cite{king2003capacity} it is shown that for $\liou_{\frac{\one}{d}}(X)=\tr(X)\frac{\one}{d}-X$ we have 
\begin{align}\label{equ:capdep}
C\lb e^{t\liou_{\frac{\one}{d}}}\rb=\log(d)+\lb e^{-t}+c(t,d)\rb\log\lb e^{-t}+c(t,d)\rb+(d-1)c(t,d)\log\lb c(t,d)\rb 
\end{align}
with $c(t,d)=(1-e^{-t})d^{-1}$.
In \cite{strongconversedep} the strong converse property was established. The semigroup generated by $\liou_{\frac{\one}{d}}$ is
therefore a natural candidate to evaluate the quality of our bounds, as determining its classical capacity can be considered a solved problem.
As $\liou_{\frac{\one}{d}}$ is just the difference of a projection and the identity, it is easy to see that the spectral gap of $\liou_{\frac{\one}{d}}$ is $1$, which gives us the upper bound
\begin{align}\label{equ:capdepbound}
C\lb e^{t\liou_{\frac{\one}{d}}}\rb\leq \log(d)e^{-\frac{(1-2d^{-2})}{2(\ln(3)\ln(d^2-1)+2(1-2d^{-2})}t}
\end{align}
for $d>2$ and
\begin{align}
C\lb e^{t\liou_{\frac{\one}{2}}}\rb \leq e^{-\frac{t}{2}} 
\end{align}
for $d=2$.
\begin{figure}\label{fig:capdep}
\begin{center}
\includegraphics[width=.4\textwidth]{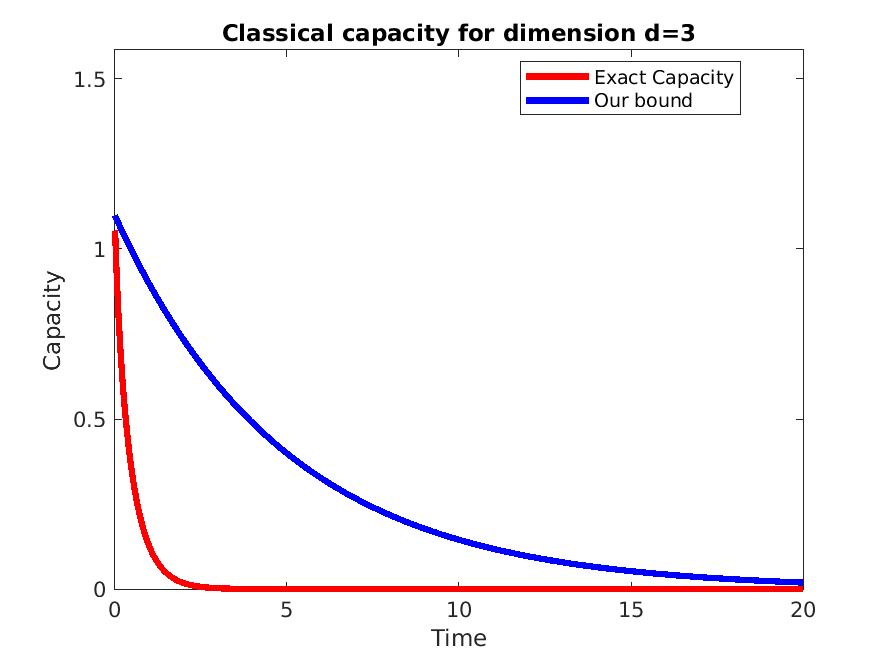}
\includegraphics[width=.4\textwidth]{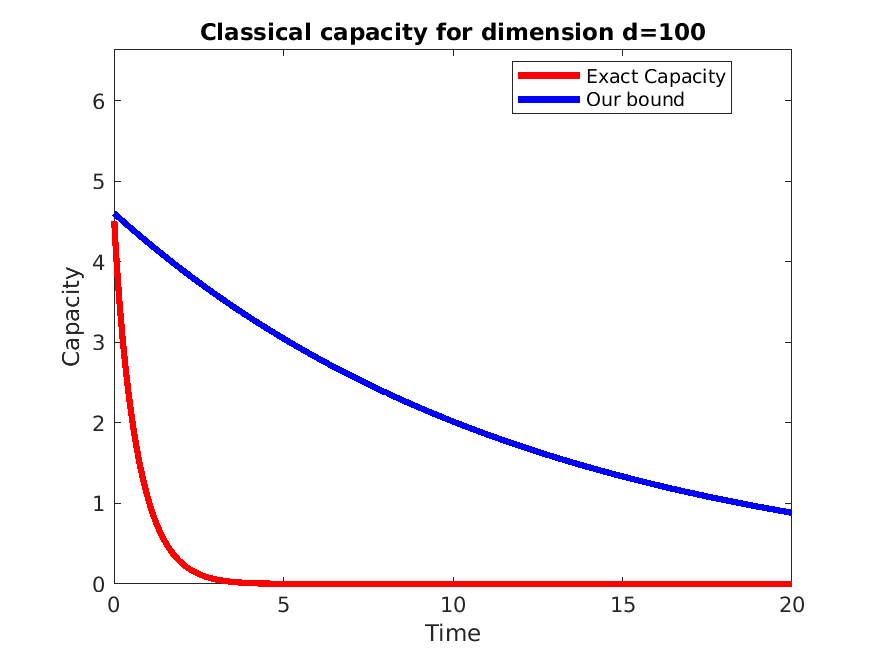}
\end{center}
\caption{Comparison of the capacity of the depolarizing channel, given in Equation \eqref{equ:capdep},  and the bound obtained by our methods, given in Equation \eqref{equ:capdepbound}. }
\end{figure}

\subsection{Stabilizer Hamiltonians}\label{sec:stabilizerhamilt}
Estimates on the spectral gap of Davies generators of stabilizer Hamiltonians were obtained in~\cite{Temme2017}. 
In the following we will make the same assumptions as in \cite{Temme2017} on the coupling of the system to the bath. That is, we assume that 
the operators $S^\alpha$ (see \eqref{equ:couplingdavies}) are given by single qubit 
Pauli operators $\sigma_x,\sigma_y$ or $\sigma_z$. For the transition rates $G^\alpha(\omega)$ we only assume that they satisfy the KMS condition~\cite{Kossakowski1977}, that is,
 $G^\alpha(-\omega)=G^\alpha(\omega)e^{-\beta\omega}$. This condition implies that the semigroup is reversible. Recall that for Davies
generators at inverse temperature $\beta>0$, which we will denote by $\liou_\beta$, the stationary state is always given by the thermal state $\frac{e^{-\beta H}}{\tr\lb e^{-\beta H}\rb}$. 

We will not discuss stabilizer Hamiltonians and groups and their connection to error-correcting codes, but refer to~\cite[Section 10.5]{nielsen2000quantum}
for more details. Given some stabilizer group $S\subset\mathcal{P}_n$, where $\mathcal{P}_n$ is the group generated by the tensor product of $n$
Pauli matrices, with commutative generators $S=<P_1,\ldots,P_k>$, we define the stabilizer Hamiltonian to be given by
\begin{align*}
H_S=-\sum\limits_{i=1}^kP_i. 
\end{align*}
We then have:
\begin{lem}\label{lem:boundsmalleststab}
Let $H_S$ be the stabilizer Hamiltonian of the stabilizer group $S=<P_1,\ldots,P_k>$ 
on $n$ qubits. Denote by $\sigma_\beta=\frac{e^{-\beta H_S}}{\tr \left(e^{-\beta H_S}\right)}$ the corresponding thermal state at inverse inverse temperature $\beta>0$. Then
\begin{align*}
\|\sigma_\beta^{-1}\|_\infty\leq2^{n}e^{2k\beta}. 
\end{align*}
\end{lem}
\begin{proof}
The eigenvalues of each $P_i$ are contained in $\{1,-1\}$,
as they are just tensor products of Pauli matrices.
From this we have
\begin{align}\label{equ:spectrumstabilizer}
-k\one\leq H_S\leq k\one,
\end{align}
as $H_S$ is just the sum of $k$ terms such that $-\one\leq P_i\leq \one$.
From \eqref{equ:spectrumstabilizer} it follows that
\begin{align}\label{equ:boundpartition}
\tr\lb e^{-\beta H_S}\rb\leq2^ne^{\beta k}, 
\end{align}
as we have $2^n$ eigenvalues, including multiplicities. Moreover, it also follows that
\begin{align}\label{equ:boundlargesthamilt}
\|e^{\beta H_S}\|_\infty\leq e^{\beta k}. 
\end{align}
As $\|\sigma_\beta^{-1}\|_\infty=\|e^{\beta H_S}\|_\infty\tr\lb e^{-\beta H_S}\rb$, the claim follows by putting \eqref{equ:boundpartition}
and \eqref{equ:boundlargesthamilt} together.
\end{proof}

In~\cite[Theorem 15]{Temme2017} they show
\begin{align*}
\lambda\geq\frac{h^*}{4\eta^*} e^{-2\beta \bar{\epsilon}}
\end{align*}
for the spectral gap $\lambda$ of the Davies generators of stabilizer Hamiltonians at inverse temperature $\beta>0$.
Here $\bar{\epsilon}$ is the generalized energy barrier, $h^*$ is the smallest transition rate and $\eta^*$ the longest path in Pauli space.
We refer to \cite{Temme2017} for the exact definition of  these parameters. It is important to stress that in general $\eta^*$ will scale with the number 
of qubits, so our estimate on the capacity will not be very good as the number of qubits increases.

However, in~\cite[Theorem 15]{Temme2017} they also show the estimate 
\begin{align*}
\lambda\geq\frac{h^*}{4} e^{-2\beta \bar{\epsilon}},
\end{align*}
for the special case in which the generalized energy barrier can be evaluated with canonical paths $\Gamma_1$. We again refer to \cite{Temme2017} for the exact
definition. For these cases the gap does not scale with the dimension and our estimate is much better.
Summing up we obtain:
\begin{thm}\label{thm:equlifetimestabilizer}
Let $H_S$ be the stabilizer Hamiltonian of the stabilizer group $S=<P_1,\ldots,P_k>$ 
on $n$ qubits. Moreover, let $\liou_\beta$ be its Davies generator
at inverse temperature $\beta>0$. Then the classical capacity $C(e^{t\liou_\beta})$ is bounded by
\begin{align*}
C(e^{t\liou_\beta})\leq\left(n+2\beta k\log(e)\right)e^{-r(\beta,n,k)t}, 
\end{align*}
with
\begin{align*}
r(\beta,n,k)=e^{-2\beta \bar{\epsilon}}\frac{h^*}{8\eta^*\lb2k\beta+5n\ln\lb2\rb+11\rb} 
\end{align*}
and
\begin{align*}
r(\beta,n,k)=e^{-2\beta \bar{\epsilon}}\frac{h^*}{8\lb2k\beta+5n\ln\lb2\rb+11\rb} 
\end{align*}
in case the generalized energy barrier can be evaluated with canonical paths $\Gamma_1$.
Moreover, this is a bound in the strong converse sense.
\end{thm}
\begin{proof}
The claim follows immediately after inserting the bounds from Lemma \ref{lem:boundsmalleststab} and 
\cite[Theorem 15,16]{Temme2017} into Corollary \ref{cor:decaycap}. 
\end{proof}
In \cite{Temme2017} one can find more explicit bounds for the parameters $\bar{\epsilon}$, $\eta^*$ and $h^*$ for some stabilizer groups.
To the best of our knowledge this is the first bound available for the classical capacity of this class of quantum channels.
To make the bound in Theorem \ref{thm:equlifetimestabilizer} more concrete, we show what we obtain for the $2D$ toric code.

\subsection{2D Toric Code}
Here we consider the 2D toric code as originally introduced in~\cite{bravyi1998quantum}, which
is a stabilizer code.
We consider only square lattices: We take an $N\times N$ lattice with $N^2$ vertical and $(N+1)^2$ horizontal edges; 
associating a qubit to each edge gives a total of $n=2N^2+2N+1$ physical qubits. 
The stabilizer operators are $N(N+1)$ plaquette operators (including the ``open'' plaquettes along the rough boundary) and $N(N+1)$ vertex operators, all of which are independent. 
It goes beyond the scope of this article to explain the $2D$ toric code in detail and we refer to~\cite[Section 19.4]{lidar2013quantum} for a discussion. 
But from the previous observations
we obtain that we have $k=2N(N+1)$ generators for the stabilizer group of the $2D$ toric code on $n=2N^2+2N+1$ qubits. 
We will make the same assumptions on the the Davies generators at inverse temperature $\beta>0$ for the toric code as in \cite{Temme2017}. These are discussed in the beginning of Subsection \ref{sec:stabilizerhamilt}.

In \cite{alicki2009thermalization} it was proved that the spectral gap for the Davies generators for the $2D$ toric code at inverse temperature $\beta$ satisfies 
$\lambda\geq\frac{1}{3}e^{-8\beta}$, a result which was reproved in \cite{Temme2017} using different techniques.
We therefore obtain:
\begin{cor}\label{cor:boundcaptoric}
Let $H$ be the stabilizer Hamiltonian of the $2D$ toric code on a $N\times N$ lattice and $\liou_\beta$ be its Davies generator
at inverse temperature $\beta>0$. Then the classical capacity $C(e^{t\liou_\beta})$ is bounded by
\begin{align}\label{equ:boundtori}
C(e^{t\liou_\beta})\leq\left(2N^2+2N+1+\log(2)4\beta N(N+1)\right)e^{-r(\beta,L)t}, 
\end{align}
with 
\begin{align*}
r(\beta,N)=\frac{e^{-8\beta}}{6\lb(10N^2+10N+5)\ln(2)+4\beta N(N+1)\rb+66}. 
\end{align*}
Moreover, this is a bound in the strong converse sense.
\end{cor}
\begin{proof}
The claim follows immediately from Lemma \ref{lem:boundsmalleststab}
and the spectral gap estimate of \cite{alicki2009thermalization} for the toric code.
\end{proof}
\begin{figure}
\begin{center}
\includegraphics[width=.9\textwidth]{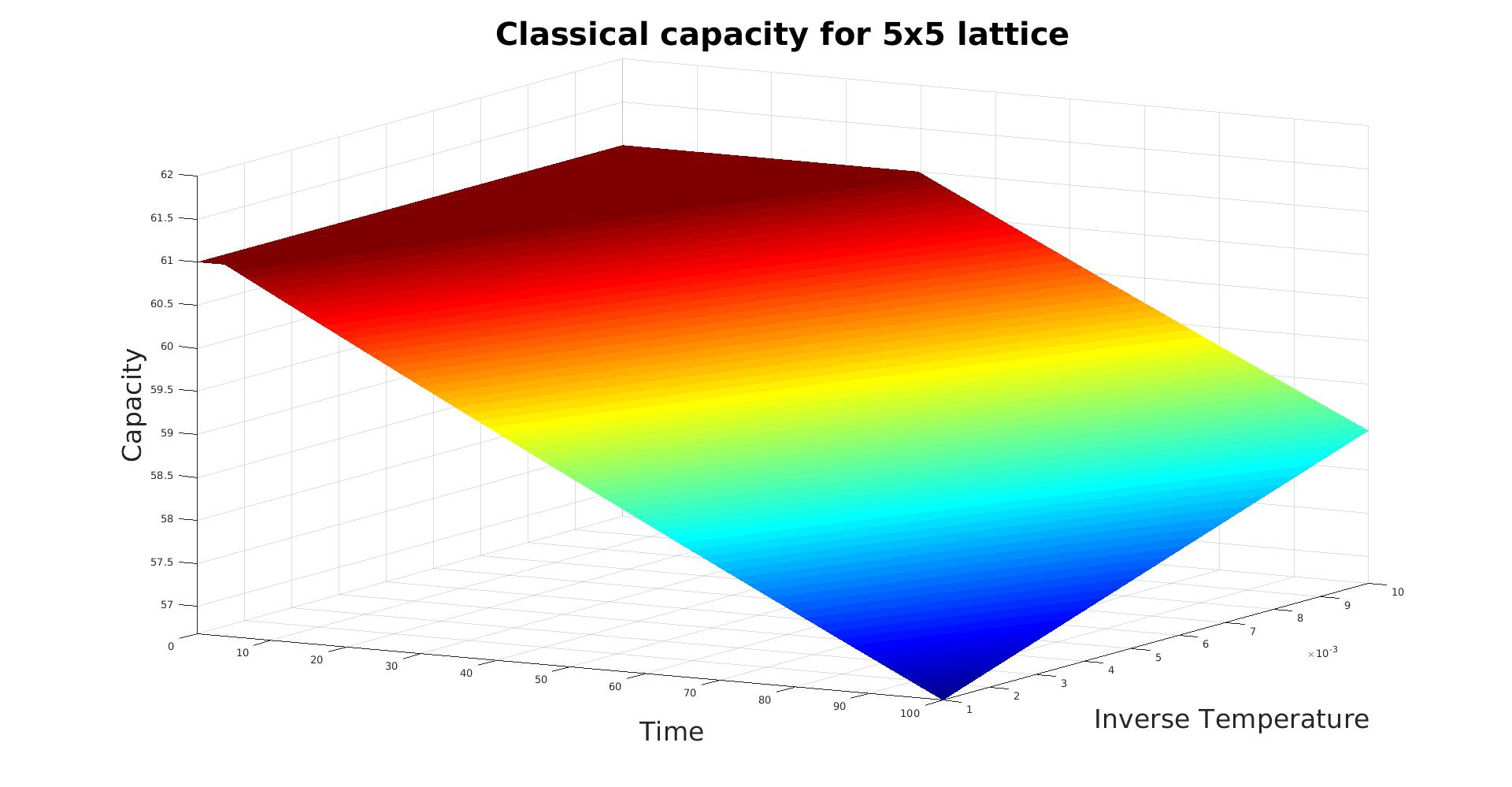}
\end{center}
\caption{Plot for a $5\times5$ lattice of the minimum of the bound in Equation \eqref{equ:boundtori}
and the trivial bound $C(e^{t\liou_\beta})\leq2N^2+2N+1=61$ as a function of the inverse temperature and time
 for the Davies generator of the $2D$ toric code.}\label{fig:captoric}
\end{figure}
From Figure \ref{fig:captoric} it becomes evident that we cannot retain information in the $2D$ toric for long times at small inverse temperatures and that we can get
nontrivial estimates even for very high dimensions, as the size of the gap does not scale with the size of the lattice.
It is conjectured that if the spectral gap of the Davies generators of a Hamiltonian with local, commuting terms satisfies
a lower bound which is independent of the size of the lattice, then the logarithmic Sobolev $2$ constant also satisfies such a bound~\cite{KastoryanoSpectralGap}.
As the Hamiltonian of the $2D$ toric code is of this form, proving this conjecture would lead to a bound similar to
the one in Corollary \ref{cor:boundcaptoric}, but with a rate $r(\beta,N)$ independent of the size of the lattice.
This would of course lead to much better bounds for large lattice sizes.

\subsection{Truncated harmonic oscillator}
Consider the Hamiltonian of a truncated harmonic oscillator
\[
H=\sum\limits_{n=0}^dn\ketbra{n}\in\M_{d+1}.\]
Suppose that the systems couples to the bath via
$S=(a+a^\dagger)$, with 
\begin{align}
a^\dagger=\sum\limits_{n=1}^d\sqrt{n}\ket{n}\bra{n-1} 
\end{align}
and the transition rate function $G(x)=(1+e^{-x\beta})^{-1}$.
Let $\sigma_\beta=\frac{e^{-\beta H}}{\tr(e^{-\beta H})}$. As the eigenvalues of $e^{-\beta H}$ are just a geometric sequence, we have 
\begin{align}
\|\sigma_\beta^{-1}\|_\infty=\frac{1-e^{-\beta(d+1)}}{1-e^{-\beta}}e^{\beta d}. 
\end{align}
In~\cite[Section V, Example 1]{lowerbounddavies} they show 
\begin{align}\label{equ:lowerbounddavies}
\lambda\geq\frac{1}{2}\min\{((1+e^{-\beta})d)^ {-1},\left[(G(1)(\sqrt{d-1}-\sqrt{d})^2+G(-1)(\sqrt{d-2}-\sqrt{d-1})^2)\right]\}, 
\end{align}
for the spectral gap $\lambda$ of the Davies generator $\liou_\beta$ of the truncated harmonic oscillator at inverse temperature $\beta>0$. 
We will denote the value of the lower bound in Equation \eqref{equ:lowerbounddavies} by $\mu(d,\beta)$.
As we can compute $\|\sigma_\beta^{-1}\|$ exactly and have a bound on the spectral gap from  we can apply Corollary \ref{cor:boundcaptoric} to these semigroups.

Note that in this case the bound scales with the dimension.
Putting these inequalities together with the bound given in Corollary \ref{cor:decaycap} for the capacity, we have for the 
classical capacity of this semigroup:
\begin{align}\label{equ:boundtrunc}
C(e^{t\liou})\leq \lb\log\left(\frac{1-e^{-\beta(d+1)}}{(1-e^{-\beta})}\right)+\beta d\log(e)\rb e^{-r(d,\beta)t},
\end{align}
with 
\begin{align}
 r(d,\beta)=\lb8\ln\lb d+1\rb+2\ln\lb\frac{1-e^{-\beta(d+1)}}{1-e^{-\beta}}\rb+2\beta d+22\rb^{-1}\mu(d,\beta).
\end{align}

\begin{figure}\label{fig:captruncharm}
\begin{center}
\includegraphics[width=.9\textwidth]{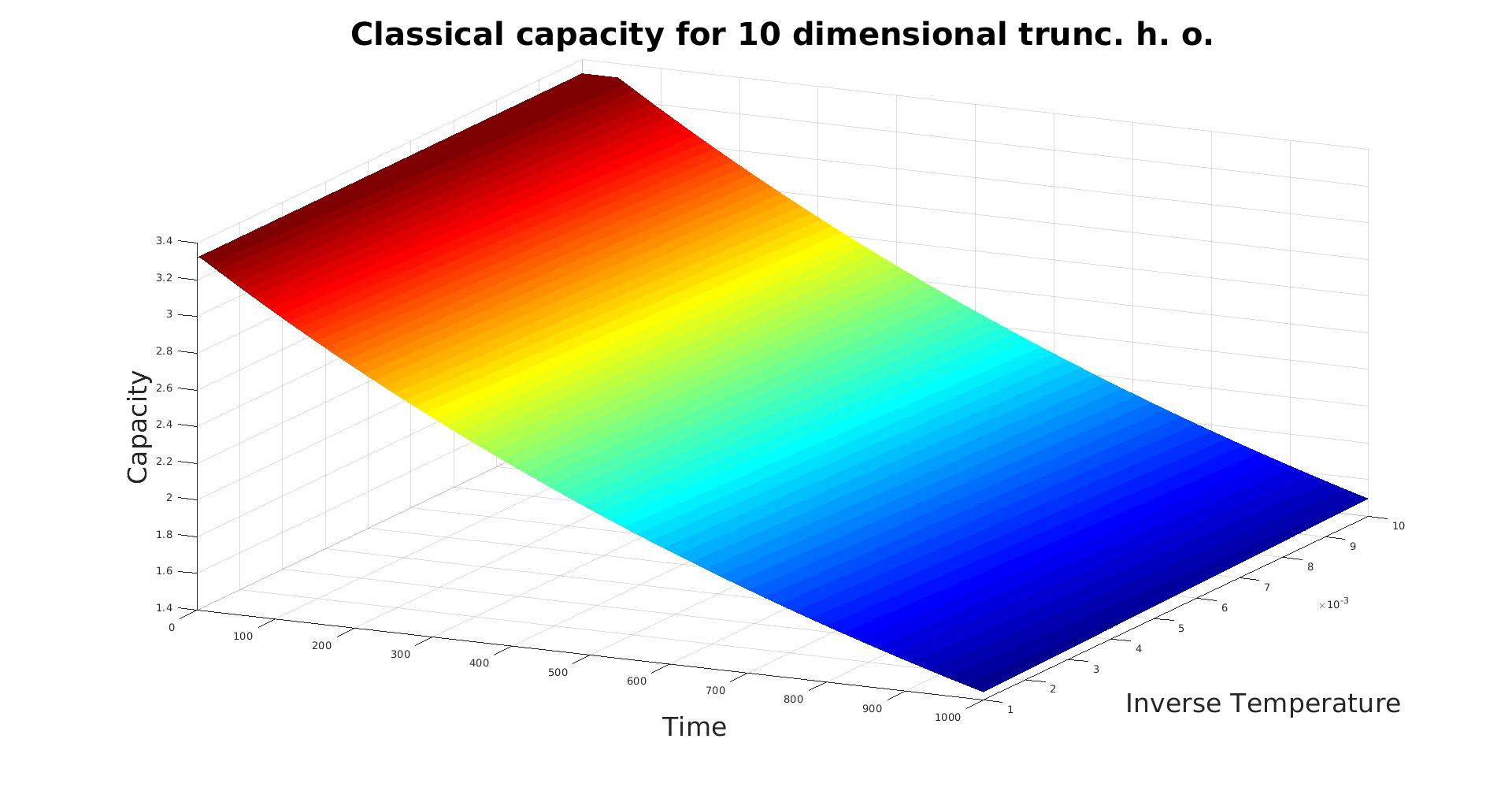}
\end{center}
\caption{Plot of the minimum of the bound in Equation \eqref{equ:boundtrunc}
and the trivial bound $C(e^{t\liou_\beta})\leq\log(10)\simeq3.32$ as a function of the inverse temperature and time for the classical capacity of the Davies generator of the truncated harmonic oscillator with $d+1=10$.}
\end{figure}

In this example we see that, as the estimate available on the gap scales with the dimension, our estimates are 
not much better than the trivial $\log(d+1)$ for high dimensions unless we are looking at large times.

\section{Conclusion and open questions}
We have introduced a framework similar to logarithmic Sobolev inequalities to study the convergence of a primitive quantum dynamical semigroup towards its fixed point in the distance measure of sandwiched R\'{e}nyi divergences. 
These techniques can be used to obtain mixing time bounds and strong converse bounds on the classical capacity of a quantum dynamical semigroup. 
Moreover, these results show that a logarithmic Sobolev inequality or hypercontractive inequality always implies a mixing time bound without the 
assumption of $l_p$-regularity (which is still not known to hold for general Liouvillians~\cite{Kastoryanosob}). 
Although we have some structural results concerning the constants $\beta_p$, some questions remain open. For logarithmic Sobolev inequalities it is known that $\alpha_2\leq\alpha_p$ for $p\geq1$ under the assumption of $l_p$-regularity (see \cite{Kastoryanosob}). It would be 
interesting to investigate if a result of similar flavor also holds for the $\beta_p$. 
In all examples discussed here, $\beta_2$ and $\alpha_2$ are of the same order and it would be interesting to know if this is always the case.
The framework of logarithmic Sobolev inequalities has recently been extended to the nonprimitive case~\cite{bardetdecoherence}. It should be possible
to develop a similar theory for the sandwiched R\'{e}nyi divergences to get rid of regularity assumptions present in their main results, as we did here
for the usual logarithmic Sobolev constants. 

We restricted our analysis to the sandwiched R\'{e}nyi divergences, as they can be expressed in terms of relative densities and noncommutative $l_p$-norms.
This allowed us to connect the convergence under the sandwiched divergences to the theory of hypercontractivity and to use tools from interpolation theory which
were vital to prove estimates on capacities.
There are however other noncommutative generalizations of the R\'{e}nyi divergences that are known to contract under quantum channels, such as 
the one discussed in~\cite[p. 113]{ohya2004quantum}. 
It would be interesting to explore the entropy production and convergence under semigroups for this and other families of divergences in future work.

In a similar vein, it would be interesting to investigate the entropy production or convergence rate for the range $\frac{1}{2}<p<1$, 
as the sandwiched R\'{e}nyi divergences
are known to contract under quantum channels for all $p>\frac{1}{2}$~\cite{liebmonotonicity}. 
However, looking closely at the proof of Theorem \ref{thm:derivativesand}, we see
that for $p<1$ the sandwiched R\'{e}nyi divergence is only differentiable at $t=0$ if the initial state has full rank. 
The  study of the convergence of these divergences for $p<1$ therefore
requires a different technical approach than that of this work.
Finally, it would of course be relevant to obtain bounds on the $\beta_p$ for more examples without relying on the estimate
based on the spectral gap, such as Davies generators.

\section*{Acknowledgments}

We thank David Reeb and Oleg Szehr for interesting and engaging discussions, which initiated this project. 
D.S.F acknowledges support from the graduate program TopMath of the Elite Network of Bavaria, the 
TopMath Graduate Center of TUM Graduate School at Technische Universit\"{a}t M\"{u}nchen, the Deutscher Akademischer 
Austauschdienst(DAAD) and by the Technische Universit\"at M\"unchen – Institute for Advanced Study,
funded by the German Excellence Initiative and the European Union Seventh Framework
Programme under grant agreement no. 291763.

A.M-H acknowledges financial support from the European Research Council (ERC Grant Agreement no 337603), 
the Danish Council for Independent Research (Sapere Aude), the Swiss National Science Foundation (project no PP00P2 150734), and the VILLUM FONDEN via the QMATH Centre of Excellence (Grant
No. 10059).

This work was supported by the German Research Foundation (DFG) and the Technical University of Munich within the funding programme Open Access Publishing.

This work was initiated at the BIRS workshop 15w5098, “Hypercontractivity and Log Sobolev Inequalities in Quantum Information Theory”. We thank BIRS and the Banff Centre for their hospitality.
\appendix

\section{Taylor Expansion of the Dirichlet Form}

In order to compute the Taylor expansions of the Dirichlet forms and of the noncommutative $l_p$-norms we define $f_p:\R^2\ra \R$ and $g_p:\R^2\ra\R$ for $p>1$ as 
\begin{equation}
f_p(x,y) = \begin{cases} (p-1)x^{p-2} & \text{ if }x=y \\
\frac{x^{p-1}-y^{p-1}}{x-y} & \text{ else}\end{cases}
\end{equation}
and 
\begin{equation}
g_p(x,y) = \begin{cases} \frac{p(p-1)}{2}x^{p-2} & \text{ if }x=y \\
\frac{(p-1)x^{p}-px^{p-1}y + y^p}{(x-y)^2} & \text{ else.}\end{cases}
\end{equation}
Note that the following identity holds 
\begin{equation}
g_p(x,y)+g_p(y,x) = pf_p(x,y)
\label{equ:Identitygf}
\end{equation}
for any $x,y\in\R$.

\begin{lem}[Taylor expansion]
Consider a primitive, reversible Liouvillian $\liou:\M_d\ra\M_d$ with full rank fixed point $\sigma\in\D^+_d$. Let $X\in\M_d$ be an eigenvector 
of $\hat{\liou} = \Gamma^{-1}_\sigma\circ \liou\circ \Gamma_\sigma$ with corresponding eigenvalue $\lambda\in\R$ (i.e.\ $\hat{\liou}(X) = \lambda X$) and $Y_\epsilon = \one_d + \epsilon X$. Then we have
\begin{equation}
\Em_p^\liou(Y_\epsilon) = \frac{p}{2(p-1)}\lb 2\epsilon^2\sum_{1\leq i\leq j\leq d}f_p(s_i,s_j)b_{ij}b_{ji} + O(\epsilon^3)\rb.
\label{equ:DirichletTaylor}
\end{equation}
and 
\begin{equation}
\kappa_p(Y_\epsilon) = \frac{p\epsilon^2}{p-1} \sum_{1\leq i\leq j\leq d}f_p(s_i,s_j)b_{ij}b_{ji} + O(\epsilon^3).
\label{equ:KappaTaylor}
\end{equation}
Where $\sigma^{1/p} = U\text{diag}\lb s_1,s_2,\ldots ,s_d\rb U^\dagger$ and $b_{ij} = (U^\dagger \sigma^{1/{2p}}X\sigma^{1/{2p}}U)_{ij}$. 
\label{lem:TaylorExp}
\end{lem}

\begin{proof}
Using that $X\in\M_d$ is an eigenvector of $\hat{\liou}$ a simple computation gives
\[
\Em_p^\liou(Y_\epsilon) = \frac{p\epsilon\sigma}{2(p-1)}\text{tr}\lb(A+\epsilon B)^{p-1} B\rb
\]
for $A = \sigma^{1/p}$ and $B = \sigma^{1/2p}X\sigma^{1/2p}$. Note that 
\[
\frac{d^k}{d\epsilon^k}\text{tr}\lb(A+\epsilon B)^{p-1} B\rb\Big{|}_{\epsilon = 0} = \text{tr}\lb D^kF(A)(B,B,\ldots ,B) B\rb
\]
for the matrix power $F:\M_d\ra\M_d$ given by $F(X)=X^{p-1}$. We apply the Daleckii-Krein formula (see \cite{daletskii1965integration} and \cite[Theorem 2.3.1.]{hiai2010matrix} for the version used here) and obtain 
\[
\frac{d}{d\epsilon}\text{tr}\lb(A+\epsilon B)^{p-1} B\rb\Big{|}_{\epsilon = 0} = \sum^d_{i=1}\sum^d_{j=1} f_p(s_i,s_j)b_{ij}b_{ji}.
\]
Using that 
\[
\text{tr}\lb(A+\epsilon B)^{p-1} B\rb\Big{|}_{\epsilon = 0} = \lk\one_d,X\rk_\sigma = 0
\]
by the orthogonality of eigenvectors, and that $f_p(x,y)=f_p(y,x)$ for any $x,y\in\R$ we obtain \eqref{equ:DirichletTaylor}.

To obtain \eqref{equ:KappaTaylor} we write 
\[
\|Y_\epsilon\|^p_{p,\sigma} = \text{tr}\lb(A+\epsilon B)^{p}\rb
\]
with $A = \sigma^{1/p}$ and $B = \sigma^{1/2p}X\sigma^{1/2p}$ as above. Again it is easy to see that 
\[
\frac{d^k}{d\epsilon^k}\text{tr}\lb(A+\epsilon B)^{p}\rb\Big{|}_{\epsilon = 0} = \text{tr}\lb D^kG(A)(B,B,\ldots ,B)\rb
\]
for the matrix power $G:\M_d\ra\M_d$ given by $G(X)=X^p$. Using the Daleckii-Krein formulas we obtain the derivatives
\begin{align*}
\frac{d}{d\epsilon}\text{tr}\lb(A+\epsilon B)^{p}\rb\Big{|}_{\epsilon = 0} &= p\lk\one_d,X\rk_\sigma=0 \\
\frac{d^2}{d\epsilon^2}\text{tr}\lb(A+\epsilon B)^{p}\rb\Big{|}_{\epsilon = 0} &= \sum^d_{i=1}\sum^d_{j=1}g_p(\sigma_i,\sigma_j)b_{ij}b_{ji} \\
&= p\sum_{1\leq i\leq j\leq d} f_p(\sigma_i,\sigma_j)b_{ij}b_{ji}
\end{align*}
where we used the identity \eqref{equ:Identitygf} in the last step. The above shows that 
\begin{equation}
\|Y_\epsilon\|^p_{p,\sigma} = 1 + \epsilon^2 p\sum_{1\leq i\leq j\leq d} f_p(\sigma_i,\sigma_j)b_{ij}b_{ji} + O(\epsilon^3).
\end{equation}
With the well-known expansion $\ln(1+x) = x - \frac{x^2}{2} + O(x^3)$ we obtain
\begin{align*}
\kappa_p(Y_\epsilon) &= \kappa_p(Y_\epsilon) = \frac{p\epsilon^2}{p-1} \sum_{1\leq i\leq j\leq d}f_p(\lambda_i,\lambda_j)b_{ij}b_{ji} + O(\epsilon^3).
\end{align*}
which is \eqref{equ:KappaTaylor}.

\end{proof}

\section{Interpolation Theorems and Proof of Theorem \ref{thm:uncertaintytau2ls2}}\label{sec:l2uncertproof}

In order to prove Theorem \ref{thm:uncertaintytau2ls2} we will need the following special case of the Stein-Weiss interpolation theorem~\cite[Theorem 1.1.1]{bergh2012interpolation}. 
This classic result from interpolation  spaces has been applied recently to solve problems from quantum information theory, such as in \cite[Section III]{dataprocessingbeigi}.

\begin{thm}[Hadamard Three Line Theorem]\label{steinweiss}
 Let $S=\{z\in\mathbb{C}:0\leq z\leq1\}$ and $F:S\to \mathcal{B}\lb \M_d\rb $ be an operator-valued function holomorphic in the interior of $S$ and uniformly bounded and continuous on the boundary. 
 Let $\sigma\in\D_d^+$ and assume $1\leq p\leq q\leq \infty$. For $0<\theta<1$
 define $p_0\leq p_\theta\leq p_1$ by
 \[
  \frac{1}{p_\theta}=\frac{1-\theta}{p_0}+\frac{\theta}{p_1}
 \]
Then for $0\leq y\leq x\leq1$  we have
\begin{equation}
\| F\lb y\rb \|_{2\to p_\theta,\sigma}\leq\sup_{a,b\in\R}\| F\lb ia\rb \|^{1-\theta}_{2\to p_0,\sigma}\| F\lb x+ib\rb \|^{\theta}_{2\to p_1,\sigma}
\end{equation}
\end{thm}

One important consequence of the Stein-Weiss interpolation theorem is the following interpolation result. We again refer to \cite[Theorem 1.1.1]{bergh2012interpolation} for a proof.
\begin{thm}[Riesz-Thorin Interpolation Theorem]\label{thm:rieszthorin}
Let $L:\M_d\to\M_d$ be a linear map, $1\leq p_0\leq p_1\leq+\infty$ and $1\leq q_0\leq q_1\leq+\infty$.
For $\theta\in[0,1]$ define $p_\theta$ to satisfy
\[
 \frac{1}{p_\theta}=\frac{\theta}{p_0}+\frac{1-\theta}{p_1}
\]
and $q_\theta$ analogously.
Then for $\sigma\in\D_d^+$ we have:
\[
\|L\|_{p_\theta\to q_\theta,\sigma}\leq\|L\|_{p_0\to q_0,\sigma}^\theta\|L\|_{p_1\to q_1,\sigma}^{1-\theta} 
\]

\end{thm}

With these tools at hand we can finally prove Theorem \ref{thm:uncertaintytau2ls2}:
\begin{proof}
 Define $E:\M_d\to\M_d$ by $E\lb X\rb =\tr\lb \sigma X\rb \one_d$ and set $\tau=t_2\lb \epsilon\rb$ for some $\epsilon>0$. 
 In the following we use $T_z=e^{\tau z\liou}$ for $z\in S = \{z\in\mathbb{C}:0\leq \text{Re}z\leq1\}$. We will show that for $s\in[0,1]$:
 \begin{equation}
  \|\lb T_{s}-E\rb \lb X\rb \|_{\frac{2}{1-s},\sigma}\leq \epsilon^s\|X\|_{2,\sigma}.
 \end{equation}
The family of operators $T_z-E$ clearly satisfies the assumptions of the Stein-Weiss interpolation theorem. We therefore have
\begin{equation}\label{equ:sobmixing1}
 \|T_{s}-E\|_{2\to\frac{2}{1-s},\sigma}\leq\sup_{a,b\in\R}\|T_{ia}-E\|^{1-s}_{2\to2,\sigma}
 \|T_{1+ib}-E\|^s_{2\to\infty,\sigma}.
\end{equation}
Observe that by reversibility of $\liou$ the map $T_{ia}$ is a unitary operator with respect to $\lk\cdot , \cdot\rk_\sigma$. We also have $T_{ia}\circ E=E$, 
as $T_{ia}(\one_d)=\one_d$. This gives
\[
\|\lb T_{ia}-E\rb \lb X\rb \|_{2,\sigma}=\|T_{ia}\lb X-E\lb X\rb \rb \|_{2,\sigma}=\|X-E\lb X\rb \|_{2,\sigma}\leq\|X\|_{2,\sigma},
\]
where the last equality follows from $\|X-\tr\lb \sigma X\rb \one_d\|_{2,\sigma}=\min\limits_{c\in\R}\|X-c\one_d\|_{2,\sigma}$.
We therefore have 
\begin{equation}\label{equ:sobmixing2}
 ||T_{ia}-E||^{1-s}_{2\to2,\sigma}\leq1.
\end{equation}
Furthermore, by the unitarity of $T_{ib}$ we can compute 
\[
\|\lb T_{1+ib}-E\rb \lb X\rb \|_{\infty,\sigma}=\|T_{ib}\circ \lb T_{1}-E\rb \lb X \rb \|_{\infty,\sigma} \leq\| T_{1}-E\|_{2\to\infty,\sigma}\|X\|_{2,\sigma}
\]
Using duality of the norms and that both $T_1$ and $E$ are self-adjoint we have
\begin{equation}\label{equ:sobmixing3}
\| T_{1+ib}-E \|_{2\ra \infty,\sigma}\leq \|T_{1}-E\|_{2\to\infty,\sigma}=\| T_{1}-E\|_{1\to 2,\sigma} = \epsilon
\end{equation}
using the definition of $\tau$ in the last equality. Inserting \eqref{equ:sobmixing2} and \eqref{equ:sobmixing3} into \eqref{equ:sobmixing1} we get
\begin{equation}\label{equ:ineqnorm}
  \epsilon^{-s}\|\lb T_{s}-E\rb \lb X\rb \|_{\frac{2}{1-s},\sigma}\leq \|X-E\lb X\rb\|_{2,\sigma}, 
\end{equation}
as $\|\lb T_{s}-E\rb \lb X\rb \|_{\frac{2}{1-s},\sigma}=\|\lb T_{s}-E\rb \lb X-E\lb X\rb\rb \|_{\frac{2}{1-s},\sigma}$. 

Taking the derivative of \eqref{equ:ineqnorm} with respect to $s$ on both sides at $s=0$ we get 
\begin{equation}\label{equ:ineqbeforerot1}
 \frac{1}{2\|X-E\lb X\rb \|_{2,\sigma}}\lb -2\|X-E\lb X\rb \|_{2,\sigma}^2\ln\lb \epsilon\rb +\Ent_{2,\sigma}\lb |X-E\lb X\rb |\rb -2\tau\mathcal{E}\lb X\rb \rb \leq0.
\end{equation}
Rearranging the terms in \eqref{equ:ineqbeforerot1} we obtain
\begin{equation}\label{equ:ineqbeforerot2}
\Ent_{2,\sigma}\lb |X-E\lb X\rb |\rb \leq 2\tau\mathcal{E}\lb X\rb +2\text{Var}\lb X\rb \ln\lb \epsilon\rb. 
\end{equation}
In \cite[Theorem 4.2]{Olkiewicz1999246} the following inequality (known as Rothaus' inequality) was shown
\begin{equation}\label{equ:rothaus}
 \Ent_{2,\sigma}\lb X\rb \leq\Ent_{2,\sigma}\lb |X-E\lb X\rb |\rb +2\text{Var}\lb X\rb. 
\end{equation}
Combining inequalities \eqref{equ:rothaus} with \eqref{equ:ineqbeforerot2} and setting $\epsilon=\frac{1}{e}$
we get $t_2\lb \frac{1}{e}\rb \alpha_2\lb \liou\rb \geq\frac{1}{2}$ by the definition of the LS constant.

To prove that the inequality is tight, consider the depolarizing Liouvillian $\liou_\sigma(X)=\tr\lb X\rb\sigma-X$ for some full rank $\sigma\in\D_d^+$.
It is easy to see that $\Var_\sigma\lb e^{t\hat{\liou_\sigma}}X\rb=e^{-t}\Var_\sigma(X)$ and so $t_2\lb e^{-1}\rb=1+\ln\lb \|\sigma^{-1}\|_\infty-1\rb$. 
Restricting to operators commuting with $\sigma$, it follows from \cite[Theorem A.1]{diaconis1996} that 
\[
\alpha_2\lb\liou_\sigma\rb\leq(1-2 \|\sigma^{-1}\|_\infty^{-1})\frac{1}{2\ln\lb \|\sigma^{-1}\|_\infty-1\rb}.
\]
Thus, for a sequence $\sigma_n\in\D_d^+$ converging to a state that is not full rank we have 
\[
\lim\limits_{n\to\infty}t_2\lb e^{-1}\rb \alpha_2\lb \liou_{\sigma_n}\rb=\frac{1}{2}\].
\end{proof}

\bibliographystyle{unsrtnat}
\bibliography{bibliography}

\end{document}